\documentclass[12pt,a4paper]{article}
\usepackage{amsfonts}
\usepackage{graphicx}
\usepackage{amsmath, amssymb, eepic, graphics, color, mathrsfs}

\newtheorem{theorem}{Theorem}

\def\qed{\ifhmode\unskip\nobreak\fi\hfill
  \ifmmode\square\else$\square$\fi}
\newenvironment{proof}[1][]{\par \noindent {\bf Proof#1}.\ }{\hfill$\Box$
\par \vspace{11pt}}

\newcommand{\inst}[1]{$^{#1}$}

\DeclareSymbolFont{eulerletters}{U}{eur}{m}{n}%
\DeclareMathSymbol{\PowersetSym}{\mathord}{eulerletters}{"7D}%

\renewcommand{\epsilon}{\varepsilon}

\begin{document}

\date{}
\title{Computing Principal Components Dynamically }

\maketitle

\begin{center}
\author{\small Darko Dimitrov \inst{1}, \small Mathias Holst \inst{2}, \small Christian Knauer \inst{1}, \small Klaus Kriegel \inst{1}}

\vspace{1cm}

{\footnotesize
\inst{1} Institut f\"ur Informatik, Freie Universit\"{a}t Berlin,\\
         Takustra{\ss}e 9, D--14195 Berlin, Germany \\
         \texttt{\{darko,knauer,kriegel\}@inf.fu-berlin.de}\\
         
\inst{2} Institute of Computer Science, Universit{\"a}t Rostock,\\
         Albert Einstein Str. 21, D-18059 Rostock, Germany \\
         \texttt{mholst@informatik.uni-rostock.de}\\
} \end{center}

%
%
%
%

\begin{abstract}
In this paper we present closed-form solutions for efficiently updating the
principal components of a set of $n$ points, 
when $m$ points are added or deleted from the point set.
For both operations performed on a discrete point set in $\mathbb{R}^d$, 
we can compute the new principal components in $O(m)$ time for fixed $d$.
This is a significant improvement over the commonly used approach of recomputing
the principal components from scratch, which takes $O(n+m)$ time.
An important application of the above result is the dynamical computation
of bounding boxes based on principal component analysis. 
PCA bounding boxes are very often used
in many fields, among others in computer graphics for collision detection and fast rendering.
We have implemented and evaluated few algorithms for computing dynamically PCA bounding boxes in $\mathbb{R}^3$.
In addition, we present closed-form solutions for computing dynamically principal components
of continuous point sets in $\mathbb{R}^2$ and $\mathbb{R}^3$.
In both cases, discrete and continuous, to compute the new principal components,
no additional data structures or storage are needed.
\end{abstract}

%
%
%
%

\section{Introduction}\label{sec:intro}

{\it Principal component analysis} (PCA) \cite{j-pca-02} is probably the oldest and 
best known of the techniques of multivariate analysis.
The central idea and motivation of PCA 
is to reduce the dimensionality of a point set by identifying
{\it the most significant directions (principal components)}. 
Let $P=\{ \vec{p}_1, \vec{p}_2,\dots ,\vec{p}_n \}$ be a set of vectors (points)
in $\mathbb{R}^d$, and
$\vec{\mu}=(\mu_1, \mu_2,\dots ,\mu_d) \in \mathbb{R}^d$ be the center of gravity of $P$.
For $1 \leq k \leq d$, we use $p_{i,k}$ to denote the $k$-th coordinate of the
vector $p_i$.
Given two vectors $\vec{u}$ and $\vec{v}$, we use $\langle \vec{u},\vec{v} \rangle$ to denote 
their inner product.
For any unit vector $\vec{v} \in \mathbb{R}^d$, the {\it{variance of $P$ 
in direction $\vec{v}$}} is
\begin{equation}
{\mbox{var}}(P,\vec{v}) = \frac{1}{n}\sum_{i=1}^n \langle p_i-\vec{\mu} \, , \, \vec{v} \rangle ^2.
\end{equation}

The most significant direction corresponds to the unit vector $\vec{v}_1$ such 
that 
${\mbox{var}}(P, \vec{v}_1)$ is maximum. In general, after identifying the $j$ most 
significant directions $\vec{v}_1, \dots , \vec{v}_j$, the $(j+1)$-th
most significant direction corresponds to the unit vector $\vec{v}_{j+1}$ such that 
${\mbox{var}}(P, \vec{v}_{j+1})$ is maximum among all unit vectors
perpendicular to $\vec{v}_1, \vec{v}_2, \dots , \vec{v}_j$. 

It can be verified that for any unit vector $\vec{v} \in \mathbb{R}^d$,
\begin{equation} \label{eq:var-cov}
{\mbox{var}}(P, \vec{v}) = \langle \Sigma \, \vec{v}, \vec{v} \rangle ,
\end{equation}
where $\Sigma$ is the {\it covariance matrix} of $P$. 
$\Sigma$ is a symmetric $d \times d$ matrix
where the ($i,j$)-th component, $\sigma_{ij}, 1 \leq i,j \leq d$, is defined as
\begin{equation}
\sigma_{ij} = \frac{1}{n}\sum_{k=1}^n (p_{i,k} - \mu_i)(p_{j,k} - \mu_j).
\end{equation}

The procedure of finding the most significant directions, in the sense 
mentioned above, can be formulated as an eigenvalue problem. If 
$\lambda_1 \geq \lambda_2 \geq \cdots \geq \lambda_d$ are the 
eigenvalues of $\Sigma$, 
then the unit eigenvector $\vec{v}_j$ for $\lambda_j$ is the
$j$-th most significant direction. 
All $\lambda_j$s are 
non-negative and $\lambda_j = {\mbox{var}}(X, \vec{v}_j)$. Since the matrix $\Sigma$ is symmetric
positive definite, its eigenvectors are orthogonal. 

Computation of the eigenvalues, when $d$ is not very large,
can be done
in $O(d^3)$ time,  for example with the {\it Jacobi} or the {\it $QR$ method}
\cite{ptvf_nrc-95}.
For very large $d$, the problem of computing eigenvalues is non-trivial.
In practice, the above mentioned methods for computing eigenvalues converge rapidly.
In theory, it is unclear how to bound the running time combinatorially and how to compute the eigenvalues
in decreasing order.  In \cite{cww-pddupca-05} a modification of the {\it Power method} \cite{Parlett-98}
is presented, which can give a guaranteed approximation of the eigenvalues with high probability.

Examples of many applications of PCA include 
data compression, exploratory data analysis, visualization, image processing,
pattern and  image recognition, time series prediction, detecting perfect
and reflective symmetry, and dimension detection.
The thorough overview over PCA's applications can be found for example in the textbooks \cite{dhs-pc-01}
and \cite{j-pca-02}.
Most of the applications of PCA are non-geometric in their nature. However, there are also 
few purely geometric applications that are quite spread in computer graphics.
Example are the estimation of the undirected normals of the point sets or computing PCA bounding boxes
(bounding boxes determined by the principal components of the point set).

Dynamic versions of the above applications, i.e., when the point set (population) changes, 
are of big importance and interest.
In this paper we present closed-form solutions for efficiently updating the
principal components of a dynamic point set.
We also consider the computation of the dynamic PCA bounding boxes -
a very important application in many fields including computer graphics, where the PCA boxes are used to 
maintain hierarchical data structures for fast rendering of a scene or for collision detection. 

Based on the theoretical results in this paper, we have implemented several algorithms for computing 
PCA bounding boxes dynamically. 

The organization and the main results of the paper are as follows: 
In Section \ref{sec:PCsDynamicaly1} 
we present closed-form solutions for efficiently updating the
principal components of a set of $n$ points, 
when $m$ points are added or deleted from the point set. 
For both operations performed on a discrete point set in $\mathbb{R}^d$, 
we can compute the new principal components in $O(m)$ time for fixed $d$.
This is a significant improvement over the commonly used approach of recomputing
the principal components from scratch, which takes $O(n+m)$ time.
In Section \ref{sec:applications} we consider solutions
for the static and dynamic versions of the bounding box problem.
In Section~\ref{sec:experiments} we present and verify the correctness of the theoretical results presented in the
Chapter~\ref{sec:PCsDynamicaly1}. We have implemented several dynamic PCA bounding box algorithms and evaluated their performances.
Conclusion and open problems are presented in Section~\ref{sec:conclusion}.
In the appendix we consider the computation of the principal components of a dynamic  continuous point set.
We give closed form-solutions when the point set is a convex polytope or a boundary of a 
convex polytope in $\mathbb{R}^2$ or $\mathbb{R}^3$.
When the point set is a boundary of a convex polytope,
we can update the new principal components in $O(k)$ time, 
for both deletion and addition,
under the assumption that we know the $k$ facets in which the polytope changes.
Under the same assumption, 
when the point set is a convex polytope in $\mathbb{R}^2$ or $\mathbb{R}^3$, 
we can update the principal components in $O(k)$ time after adding points.
But, to update the principal components after deleting points from a convex polytope in $\mathbb{R}^2$ or $\mathbb{R}^3$
we need $O(n)$ time.
This is due to the fact that after a deletion
the center of gravity of the old convex hull (polyhedron)
could lie outside the new convex hull, and therefore, a retetrahedralization is needed 
(see Subsection~\ref{subsec:cpca-polytope} and Subsection~\ref{subsec:cpca-polygon} for details).

\section{Updating the principal components efficiently - 
         \\ discrete case in $\mathbb{R}^d$}\label{sec:PCsDynamicaly1}

In this subsection, we consider the problem of updating the covariance matrix 
$\Sigma$ of a discrete point set $P=\{ \vec{p}_1, \vec{p}_2,\dots ,\vec{p}_n \}$ in $\mathbb{R}^d$, when $m$ points are added or deleted from $P$.
We give closed-form solutions for computing the components of 
the new covariance matrix $\Sigma'$. Those closed-form solutions
are based on the already computed components of $\Sigma$.
The main result of this section is given in the following theorem.

\begin{theorem} \label{thm:discret-Rd}
Let $P$ be a set of $n$ points in $\mathbb{R}^d$ with known covariance matrix $\Sigma$.
Let $P'$ be a point set in $\mathbb{R}^d$, obtained by adding or deleting $m$ points from $P$.
The principal components of $P'$ can be computed in $O(m)$ time for fixed $d$. 
\end{theorem}

\begin{proof}

\bigskip
\noindent
{\bf{Adding points}}
\bigskip

Let $P_m=\{\vec{p}_{n+1}, \vec{p}_{n+2}, \dots, \vec{p}_{n+m}\}$ be a point set with
center of gravity $\vec{\mu}^m = (\mu_1^m, \mu_2^m, \dots, \mu_d^m)$.
We add $P_m$ to $P$ obtaining new point 
set $P'$.
The $j$-th component, $\mu_j'$, $1 \leq j \leq d$, of the center of gravity 
$\vec{\mu}' = (\mu_1', \mu_2', \dots, \mu_d')$ of $P'$ is

$$
\begin{array}{lll}
\mu_j'& = & \frac{1}{n+m} \sum_{k=1}^{n+m} p_{k,j} \vspace{0.3cm} \\ 
      & = & \frac{1}{n+m} \left(\sum_{k=1}^{n}p_{k,j} + \sum_{k=n+1}^{n+m}p_{k,j}\right) \vspace{0.3cm} \\
      & = & \frac{n}{n+m}\mu_j + \frac{m}{n+m}\mu_j^m.
\end{array}
$$

The $(i,j)$-th component, $\sigma_{ij}'$, $1 \leq i, j \leq d$, of the covariance matrix 
$\Sigma'$ of $P'$ is
$$
\begin{array}{lll}
\sigma_{ij}'& = & \frac{1}{n+m} \sum_{k=1}^{n+m} (p_{k,i} - \mu_i')(p_{k,j} - \mu_j')  \vspace{0.3cm}\\
      & = & \frac{1}{n+m} \sum_{k=1}^{n} (p_{k,i} - \mu_i')(p_{k,j} - \mu_j') 
            + \frac{1}{n+m} \sum_{k=n+1}^{n+m} (p_{k,i} - \mu_i')(p_{k,j} - \mu_j'). \\
\end{array}
$$

Let

$$
\sigma_{ij}' = \sigma_{ij,1}' + \sigma_{ij,2}',
$$

where,
\begin{equation}\label{eq:add30}
\sigma_{ij,1}' = \frac{1}{n+m} \sum_{k=1}^{n} (p_{k,i} - \mu_i')(p_{k,j} - \mu_j'), 
\end{equation}

and

\begin{equation}\label{eq:add40}
\sigma_{ij,2}' = \frac{1}{n+m} \sum_{k=n+1}^{n+m} (p_{k,i} - \mu_i')(p_{k,j} - \mu_j').
\end{equation}

Plugging-in the values of $\mu_i'$ and $\mu_j'$ in (\ref{eq:add30}), we obtain:
$$
\begin{array}{lll}
\sigma_{ij,1}' & = & \frac{1}{n+m}  \sum_{k=1}^{n} (p_{k,i} - \frac{n}{n+m}\mu_i - \frac{m}{n+m}\mu_i^m)
                                            (p_{k,j} - \frac{n}{n+m}\mu_j - \frac{m}{n+m}\mu_j^m) \vspace{0.3cm} \\
        & = & \frac{1}{n+m}  \sum_{k=1}^{m} (p_{k,i} - \mu_i +\frac{m}{n+m}\mu_i - \frac{m}{n+m}\mu_i^m)
                                            (p_{k,j} - \mu_j+\frac{m}{n+m}\mu_j - \frac{m}{n+m}\mu_j^m) \vspace{0.3cm} \\
            & = & \frac{1}{n+m} \sum_{k=1}^{n}(p_{k,i} - \mu_i)(p_{k,j} - \mu_j)+
            \frac{1}{n+m}\sum_{k=1}^{n}(p_{k,i} - \mu_i)(\frac{m}{n+m}\mu_j - \frac{m}{n+m}\mu_j^m) + \vspace{0.3cm} \\
            &  &  \frac{1}{n+m}\sum_{k=1}^{n}(\frac{m}{n+m}\mu_i - \frac{m}{n+m}\mu_i^m)(p_{k,j} - \mu_j) + \vspace{0.3cm} \\
            &  & \frac{1}{n+m}\sum_{k=1}^{n} (\frac{m}{n+m}\mu_i - \frac{m}{n+m}\mu_i^m)(\frac{m}{n+m}\mu_j -      \frac{m}{n+m}\mu_j^m)\vspace{0.3cm}. \\
\end{array}
$$

Since $\sum_{k=1}^{n}(p_{k,i} - \mu_i)=0$, $1 \leq i \leq d$, we have

\begin{equation}\label{eq:add60}
\begin{array}{lll}
\sigma_{ij,1}'
            & = & \frac{n}{n+m} \sigma_{ij} +
            \frac{n m^2}{(n+m)^3}(\mu_i - \mu_i^m )(\mu_j - \mu_j^m).
\end{array}
\end{equation}

Plugging-in the values of $\mu_i'$ and $\mu_j'$ in (\ref{eq:add40}), we obtain:

$$
\begin{array}{lll}
\sigma_{ij,2}'
        & = & \frac{1}{n+m}  \sum_{k=n+1}^{n+m} (p_{k,i} - \frac{n}{n+m}\mu_i - \frac{m}{n+m}\mu_i^m)
              (p_{k,j} - \frac{n}{n+m}\mu_j - \frac{m}{n+m}\mu_j^m) \vspace{0.3cm} \\
        & = & \frac{1}{n+m}  \sum_{k=n+1}^{n+m} (p_{k,i} - \mu_i^m +\frac{n}{n+m}\mu_i^m - \frac{n}{n+m}\mu_i)
              (p_{k,j} - \mu_j^m +\frac{n}{n+m}\mu_j^m - \frac{n}{n+m}\mu_j ) \vspace{0.3cm} \\
        & = & \frac{1}{n+m} \sum_{k=n+1}^{n+m}(p_{k,i} - \mu_i^m)(p_{k,j} - \mu_j^m)+
            \frac{1}{n+m}\sum_{k=n+1}^{n+m}(p_{k,i} - \mu_i^m) \frac{n}{n+m}(\mu_j^m - \mu_j) + \vspace{0.3cm} \\
            &  &  \frac{1}{n+m}\sum_{k=n+1}^{n+m}\frac{n}{n+m}(\mu_i^m-\mu_i)(p_{k,j} - \mu_j^m) + 
             \frac{1}{n+m}\sum_{k=n+1}^{n+m} \frac{n}{n+m}(\mu_i^m-\mu_i) \frac{n}{n+m}(\mu_j^m-\mu_j).
             \end{array}
$$

Since $\sum_{k=n+1}^{n+m}(p_{k,i} - \mu_i^m)=0$, $1 \leq i \leq d$, we have

\begin{equation}\label{eq:add70}
\begin{array}{lll}
\sigma_{ij,2}' & = & \frac{m}{n+m} \sigma_{ij}^m+
             \frac{n^2 m}{(n+m)^3}(\mu_i-\mu_i^m)(\mu_j-\mu_j^m), \vspace{0.3cm}
\end{array}
\end{equation}
where

$$\sigma_{ij}^m = \frac{1}{m} \sum_{k=n+1}^{n+m}(p_{k,i} - \mu_i^m)(p_{k,j}- \mu_j^m)),
\quad 1 \leq i,j \leq d,$$ is the $i,j$-th element of the covariance matrix 
$\Sigma_m$ of the point set $P_m$.

Finally, we have

\begin{equation}\label{eq:add80}
\sigma_{ij}' = \sigma_{ij,1}' + \sigma_{ij,2}' =\frac{1}{n+m}(n \sigma_{ij} + m \sigma_{ij}^m) +
\frac{n m}{(n+m)^2}(\mu_i-\mu_i^m)(\mu_j-\mu_j^m).
\end{equation}

Note that $\sigma_{ij}^m$, and therefore $\sigma_{ij}'$, can be computed in $O(m)$ time.
Thus, for a fixed dimension $d$, the covariance matrix $\Sigma$ also can be computed in $O(m)$ time.

\bigskip
\noindent
{\bf{Deleting points}}
\bigskip

Let $P_m=\{\vec{p}_{n-m+1}, \vec{p}_{n-m}, \dots, \vec{p}_{n}\}$ be a subset of the point set P,
and let $\vec{\mu}^m = (\mu_1^m, \mu_2^m, \dots, \mu_d^m)$ be the center of gravity of $P_m$.
We subtract $P_m$ from $P$, obtaining new point set $P'$.
The $j$-th component, $\mu_j'$, $1 \leq j \leq d$, of the center of gravity 
$\vec{\mu}' = (\mu_1', \mu_2', \dots, \mu_d')$ of $P'$ is
$$
\begin{array}{lll}
\mu_j'& = & \frac{1}{n-m} \sum_{k=1}^{n-m} p_{k,j} \vspace{0.3cm} \\ 
      & = & \frac{1}{n-m} \left(\sum_{k=1}^{n}p_{k,j} - \sum_{k=n-m+1}^{n}p_{k,j}\right) \vspace{0.3cm} \\
      & = & \frac{n}{n-m}\mu_j - \frac{m}{n-m}\mu_j^m.
\end{array}
$$

The $(i,j)$-th component, $\sigma_{ij}'$, $1 \leq i, j \leq d$, of the covariance matrix 
$\Sigma'$ of $P'$ is
$$
\begin{array}{lll}
\sigma_{ij}'& = & \frac{1}{n-m} \sum_{k=1}^{n-m} (p_{k,i} - \mu_i')(p_{k,j} - \mu_j')  \vspace{0.3cm}\\
      & = & \frac{1}{n-m} \sum_{k=1}^{n} (p_{k,i} - \mu_i')(p_{k,j} - \mu_j') 
            - \frac{1}{n-m} \sum_{k=n-m+1}^{n} (p_{k,i} - \mu_i')(p_{k,j} - \mu_j'). \\
\end{array}
$$

Let

$$
\sigma_{ij}' = \sigma_{ij,1}' - \sigma_{ij,2}',
$$

where,
\begin{equation}\label{eq:delete30}
\sigma_{ij,1}' = \frac{1}{n-m} \sum_{k=1}^{n} (p_{k,i} - \mu_i')(p_{k,j} - \mu_j'), 
\end{equation}

and

\begin{equation}\label{eq:delete40}
\sigma_{ij,2}' = \frac{1}{n-m} \sum_{k=n-m+1}^{n} (p_{k,i} - \mu_i')(p_{k,j} - \mu_j').
\end{equation}

Plugging-in the values of $\mu_i'$ and $\mu_j'$ in (\ref{eq:delete30}), we obtain:

$$
\begin{array}{lll}
\sigma_{ij,1}' & = & \frac{1}{n-m}  \sum_{k=1}^{n} (p_{k,i} - \frac{n}{n-m}\mu_i + \frac{m}{n-m}\mu_i^m)
                                            (p_{k,j} - \frac{n}{n-m}\mu_j + \frac{m}{n-m}\mu_j^m) \vspace{0.3cm} \\
        & = & \frac{1}{n-m}  \sum_{k=1}^{m} (p_{k,i} - \mu_i +\frac{m}{n-m}\mu_i - \frac{m}{n-m}\mu_i^m)
                                            (p_{k,j} - \mu_j+\frac{m}{n-m}\mu_j - \frac{m}{n-m}\mu_j^m) \vspace{0.3cm} \\
            & = & \frac{1}{n-m} \sum_{k=1}^{n}(p_{k,i} - \mu_i)(p_{k,j} - \mu_j)+
            \frac{1}{n-m}\sum_{k=1}^{n}(p_{k,i} - \mu_i)(\frac{m}{n-m}\mu_j - \frac{m}{n-m}\mu_j^m) + \vspace{0.3cm} \\
            &  &  \frac{1}{n-m}\sum_{k=1}^{n}(\frac{m}{n-m}\mu_i - \frac{m}{n-m}\mu_i^m)(p_{k,j} - \mu_j) + \vspace{0.3cm} \\
            &  & \frac{1}{n-m}\sum_{k=1}^{n} (\frac{m}{n-m}\mu_i - \frac{m}{n-m}\mu_i^m)(\frac{m}{n-m}\mu_j - \frac{m}{n-m}\mu_i^m). \vspace{0.3cm} \\
\end{array}
$$

Since $\sum_{k=1}^{n}(p_{k,i} - \mu_i)=0$, $1 \leq i \leq d$, we have

\begin{equation}\label{eq:delete50}
\begin{array}{lll}
\sigma_{ij,1}'
            & = & \frac{n}{n-m} \sigma_{ij} +
            \frac{n m^2}{(n-m)^3}(\mu_i - \mu_i^m)(\mu_j - \mu_j^m).
\end{array}
\end{equation}

Plugging-in the values of $\mu_i'$ and $\mu_j'$ in (\ref{eq:delete40}), we obtain:

$$
\begin{array}{lll}
\sigma_{ij,2}'
        & = & \frac{1}{n-m}  \sum_{k=n-m+1}^{n} (p_{k,i} - \frac{n}{n-m}\mu_i + \frac{m}{n-m}\mu_i^m)
              (p_{k,j} - \frac{n}{n-m}\mu_j + \frac{m}{n-m}\mu_j^m) \vspace{0.3cm} \\
        & = & \frac{1}{n-m}  \sum_{k=n-m+1}^{n} (p_{k,i} - \mu_i^m +\frac{n}{n-m}\mu_i^m - \frac{n}{n-m}\mu_i)
              (p_{k,j} - \mu_j^m +\frac{n}{n-m}\mu_j^m - \frac{n}{n-m}\mu_j ) \vspace{0.3cm} \\
        & = & \frac{1}{n-m} \sum_{k=n-m+1}^{n}(p_{k,i} - \mu_i^m)(p_{k,j} - \mu_j^m)+
            \frac{1}{n-m}\sum_{k=n-m+1}^{n-m}(p_{k,i} - \mu_i^m) \frac{n}{n-m}(\mu_j^m - \mu_j) + \vspace{0.3cm} \\
            &  &  \frac{1}{n-m}\sum_{k=n-m+1}^{n}\frac{n}{n-m}(\mu_i^m-\mu_i)(p_{k,j} - \mu_j^m) + 
             \frac{1}{n-m}\sum_{k=n-m+1}^{n} \frac{n}{n-m}(\mu_i^m-\mu_i) \frac{n}{n-m}(\mu_j^m-\mu_j).
\end{array}
$$

Since $\sum_{k=n-m+1}^{n}(p_{k,i} - \mu_i^m)=0$, $1 \leq i \leq d$, we have

\begin{equation}\label{eq:delete60}
\begin{array}{lll}
\sigma_{ij,2}' & = & \frac{n}{n-m} \sigma_{ij}^m +
             \frac{n^2 m}{(n-m)^3}(\mu_i-\mu_i^m)(\mu_j-\mu_j^m), \vspace{0.3cm}
\end{array}
\end{equation}
where

$$\sigma_{ij}^m = \frac{1}{m} \sum_{k=n-m+1}^{n}(p_{k,i} - \mu_i^m)(p_{k,j}- \mu_j^m)),
\quad 1 \leq i,j \leq d,$$ is the $i,j$-th element of the covariance matrix 
$\Sigma_m$ of the point set $P_m$.

Finally, we have

\begin{equation}\label{eq:delete70}
\sigma_{ij}' = \sigma_{ij,1}' + \sigma_{ij,2}' =\frac{1}{n-m}(n \sigma_{ij} - m \sigma_{ij}^m) -
\frac{n m}{(n-m)^2}(\mu_i-\mu_i^m)(\mu_j-\mu_j^m).
\end{equation}

Note that $\sigma_{ij}^m$, and therefore $\sigma_{ij}'$, can be computed in $O(m)$ time.
Thus, for a fixed dimension $d$, the covariance matrix $\Sigma$ also can be computed in $O(m)$ time.
\end{proof}

As a corollary of (\ref{eq:add80}), in the case when only one point, $\vec{p}_{n+1}$, is added to a point set $P$, 
the elements of the new covariance
matrix are given by
\begin{equation}\label{eq:add85}
\sigma_{ij}' = \sigma_{ij,1}' + \sigma_{ij,2}' =\frac{n}{n+1} \sigma_{ij} + \frac{n}{(n+1)^2} (p_{n+1,i}-\mu_i)(p_{n+1,j}-\mu_j),
\end{equation}
and can be computed in $O(1)$ time.

Similarly, as a corollary of (\ref{eq:delete70}),
in the case when only one point, $\vec{p}_{e}$, is deleted from a point set $P$, the elements of the new covariance
matrix are given by

\begin{equation}\label{eq:delete80}
\sigma_{ij}' = \sigma_{ij,1}' - \sigma_{ij,2}' =\frac{m}{m-1} \sigma_{ij} - \frac{m}{(m-1)^2} ( p_{e,i} - \mu_i)(p_{e,j} - \mu_j),
\end{equation}
and also can be computed in $O(1)$ time.

The principal components of discrete point sets can be strongly influenced by point clusters \cite{dkkr-bqpcabb-09}.
To avoid the influence of the distribution of the point set, often continuous sets, especially the convex hull of a point set is considered,
which lead to so-called continuous PCA.
Computing PCA bounding boxes \cite{glm-obbtree-96}, \cite{dhkk-cfscpcabba-09}, or retrieval of {3D}-objects \cite{vsr-fmbb-01}, are typical applications where continuous  PCA are of interest.
Due to better readability and compactness of the paper, we present  the closed-form solutions for dynamic version of
continuous PCA in the appendix. There, we consider cases when the point set is a convex polytope or the boundary of a convex polytope 
in $\mathbb{R}^2$ and $\mathbb{R}^3$.

\section{An application - computing PCA bounding boxes}\label{sec:applications}

PCA is a well-established technique for dimensionality reduction and
multivariate analysis, with numerous applications in both static and dynamic context.
Examples of many applications of PCA include 
data compression, exploratory data analysis, visualization, image processing,
pattern and  image recognition, time series prediction, detecting perfect
and reflective symmetry, and dimension detection (see textbooks \cite{dhs-pc-01}
and \cite{j-pca-02} for thorough overview over PCA's applications).

In the rest of this section, we consider solutions
for the static and dynamic versions of the bounding box problem.

\subsection{Computing bounding boxes - static version}
Substituting sets of points or complex geometric shapes  with their 
bounding boxes is motivated by many applications. For example,
in computer graphics, it is used to maintain hierarchical data structures for 
fast rendering of a scene or for collision detection. 
Additional applications include those in shape analysis and shape simplification, 
or in statistics, for storing and performing range-search queries on a large
database of samples. 

Computing a minimum-area bounding rectangle 
of a set of $n$ points in $\mathbb{R}^2$ can be done in $O(n \log n)$ 
time, for example with the  rotating calipers algorithm 
\cite{t-sgprc-83}.
O'Rourke \cite{o-fmbb-85} presented a deterministic algorithm,
a rotating calipers variant in $\mathbb{R}^3$,  
for computing the minimum-volume
bounding box of a set of $n$ points in $\mathbb{R}^3$. 
His algorithm requires 
$O(n^3)$ time and $O(n)$ space.
Barequet and Har-Peled \cite{bh-eamvb-01} have contributed two algorithms with nearly linear complexity,
based on a core-set approach, 
that compute (1+$\epsilon$)-approximations of the minimum-volume bounding box of point sets in
$\mathbb{R}^3$.
The running times of their algorithms are $O(n + 1/{\epsilon}^{4.5})$ and 
$O(n \log n + n/{\epsilon}^3)$, respectively.
A further improvement to $O(n + 1/{\epsilon}^{3})$ running time can be obtained
by using a coreset of size $O(1/{\epsilon})$ by Agarwal, Har-Peled, and Varadarajan \cite{ahv-dhsrrp-04},
and Chan \cite{c-fccdaafd-06}.

Numerous heuristics have been proposed for computing a box that encloses a
given set of points. The simplest heuristic is naturally to compute the
axis-aligned bounding box of the point set. Two-dimensional variants of this
heuristic include the well-known \emph{R-tree}, the \emph{packed R-tree}
\cite{rl-dsspduprt-85}, the \emph{$R^*$-tree} \cite{bkss-R*-tree}, 
the \emph{$R^+$-tree} \cite{srf-r+tree}, etc.

A frequently used heuristic for computing a bounding box of 
a set of points is based on PCA.
The principal components of the point set define the axes of 
the bounding box. Once the directions of the axes are given,
the dimension of the bounding box is easily found by the extreme values of the
projection of the points on the corresponding axis.

\begin{figure}[h]
  \centering
  \includegraphics[scale=0.45]{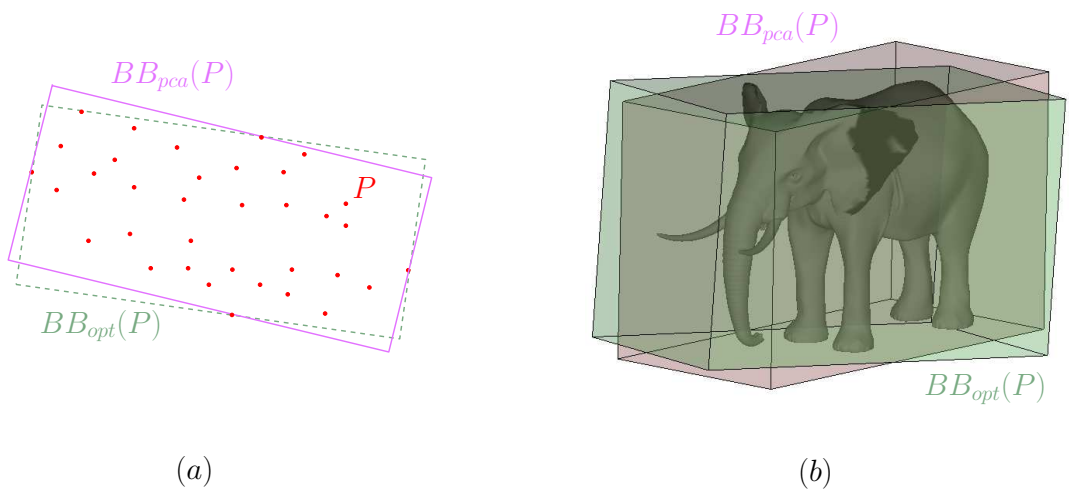}
  \caption{The minimum-area(volume) bounding box
  and the  PCA bounding box of a point set $P$, (a) in  $\mathbb{R}^2$,
  and (b) in $\mathbb{R}^3$.}
  \label{fig:pca-opt-box}
\end{figure}

Two distinguished applications of this heuristic are
the OBB-tree \cite{glm-obbtree-96} and the BOXTREE \cite{bcgmt-boxtree-96},
hierarchical bounding box structures, that support efficient collision detection
and ray tracing.
Computing a bounding box of a set of points in $\mathbb{R}^2$ 
and $\mathbb{R}^3$ by PCA is simple and 
requires linear 
time. To avoid the influence of the distribution of 
the point set 
on the directions of the PCs, a possible approach 
is to consider the convex hull, or 
the boundary of the convex hull $CH(P)$ of the point set $P$. 
Thus, the complexity of the algorithm increases to $O(n \log n)$. 
The popularity of this heuristic, besides its speed,
lies in its easy implementation and in the fact that usually 
PCA bounding boxes are tight-fitting,
see Figure~\ref{fig:pca-opt-box} for an illustration.
Experimental results of the quality of the PCA bounding boxes can be found in \cite{dhkk-cfscpcabba-09},\cite{lkmkbz-obbtdtpb-00},
and theoretical results in \cite{dkkr-bqpcabb-09}.

\subsection{Computing bounding boxes - dynamic version}

Dynamic $(1+\epsilon)$-approximation bounding box can be efficiently solved by the dynamic data structure \cite{c-dc-09} that can maintain an $\epsilon$-coreset of $n$ points in $O(\log n)$ time for any constant $\epsilon > 0$ and any constant dimension.
That is an improvement of the previous method by Agarwal, Har-Peled, and Varadarajan \cite{ahv-dhsrrp-04} that requires polylogarithmic update time. However, both results are more of theoretical importance, 
since there realization involves sophisticated data structures difficult for implementation.

Computing PCA bounding boxes of a point set consists of two steps: $1)$ computing the
principal components, that define the axes of the bounding box, and 2)
computing the extremal point along the axes, that determine the size of
the bounding box.
In Section~\ref{sec:PCsDynamicaly1}, we have presented closed-form solution for efficient update
of the principal components when we add or delete several points from the point set.
In sequel, we consider step $2$.

\subsubsection{Computing extremal points}

A trivial way to compute the extremal points is to scan all points, which takes linear time.

Faster algorithms that compute extremal points dynamically are related to computing and maintaining the 
convex hull of the point set dynamically. Then, one can perform extreme point queries in polylogarithmic time.
Here, we give an overview of the known results of dynamical computation of convex hull and some of related operation
on it in $\mathbb{R}^2$ and $\mathbb{R}^3$.

Brodal and Jacob \cite{bj-dpch-02} 
present a data structure that maintains a finite set of $n$ points in the plane under point insertions  and point deletions in amortized $O(\log n)$ time per operation. This data structure requires $O(n$) space, and supports extreme point queries in a given direction, tangent queries through a given point, and queries for the neighboring points on the convex hull in $O(\log n)$ time. 
T.~Chan \cite{c-dds3dch2dnnq-06} presents a fully dynamic randomized data structure that can answer queries about the convex hull of a set of $n$ points in three dimensions, where insertions take $O(\log^3 n)$ expected amortized time, deletions take $O(\log^6 n)$ expected amortized time, and extreme-point queries take $O(\log^2 n)$ worst-case time. This is the first method that guarantees polylogarithmic update and query cost for arbitrary sequences of insertions and deletions, and improves the previous $O(n^\epsilon)$-time method by Agarwal and Matou{\v s}ek~\cite{am-dhsrrp-95}. 

There are two disadvantages in the above approaches. First, they require a computation of the convex hull, 
which increases the complexity of the algorithms to $O(n \log n)$. However, the convex hull computation
and building corresponding data structures can be done in preprocessing, which is not critical for many applications.
Second, the above  date structures for dynamic convex hull computation are of theoretical importance, they 
are quite difficult for implementation, and to best of our knowledge, they have never been implemented.
Therefore, in the next section, we consider two simple approaches for computing extremal points, one
is the already mentioned linear scan of all points, and the other is a grid approach, refined with several variants.

\section{Practical variants of dynamical PCA bounding boxes and experimental results}\label{sec:experiments}

The main focus in this section is to show the advantages of the theoretical results presented
in this paper in the context of computing dynamic PCA bounding boxes.
We present three practical simple algorithms, and compare their performances.
A thorough comparison study of different variants of statical PCA bounding boxes
the interested reader could find in \cite{dhkk-cfscpcabba-09}.
The algorithms were implemented in C\#, C++ and OpenGL, and tested on a Core Duo 2.33GHz 
with 2GB memory. 
The principal components of all algorithms are computed with the closed-form solutions from Section~\ref{sec:PCsDynamicaly1}.
They differ only how the extremal points along the principal components are found.
The implemented algorithms are the following:

\begin{itemize}
\item  {\bf PCA-AP} (PCA-all-points) - finds the extremal points by going through all points.
\end{itemize}

\begin{itemize}
\item {\bf PCA-AGP} (PCA-all-grid-points) -the space is discretized by a regular three dimensional axis-aligned grid, 
with a cube of size $\epsilon \times \epsilon \times \epsilon$ as primer component.
See Figure~\ref{fig:lion} for an illustration. The grid size is chosen relatively to the size of the object. Each object is scaled such that its diameter is $1$.
The values of $\epsilon$ are between $0.001$ and $1$.
The corners of non-empty cells are considered to find the extremal points along the principal directions.
\begin{figure}[h]
  \centering
  \includegraphics[scale=0.45]{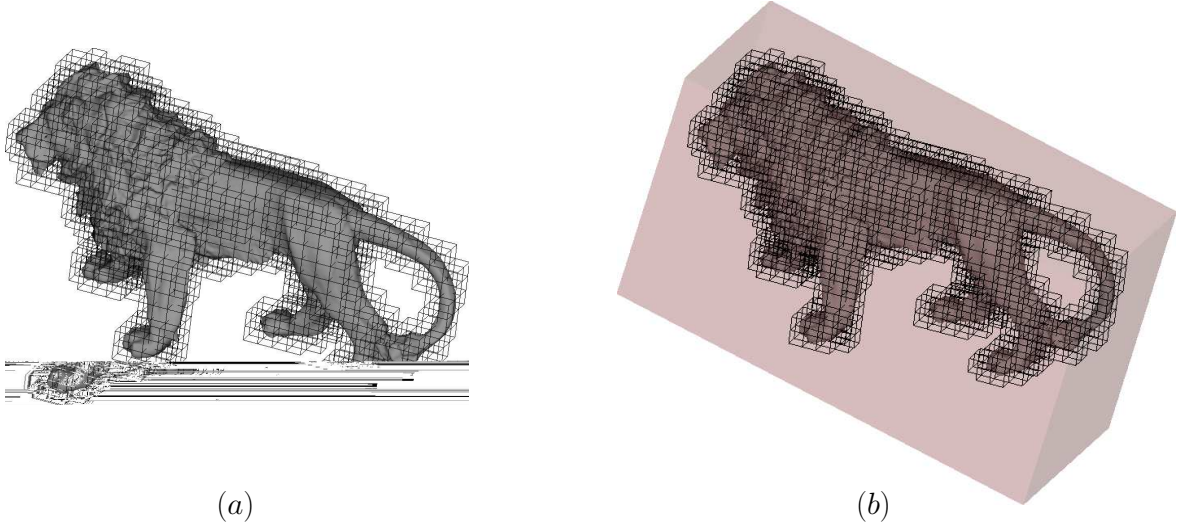}
  \caption{(a) A real world object and its corresponding grid for $\epsilon=0.03$.
  Only the non-empty grids are visualized.
  (b) The bounding box of the object obtained by the PCA-AGP algorithm.}
  \label{fig:lion}
\end{figure}
\item {\bf PCA-EGP} (PCA-extremal-grid-points) - this is an improvement of the PCA-AGP algorithm. To each vertical grid line, i.e., orthogonal to $XY$ plane, two extremal corners of the non-empty cell are computed. Thus, we reduced the candidates for extremal points from $O(\frac{1}{\epsilon^3})$ to $O(\frac{1}{\epsilon^2})$.
\end{itemize}

We further reduce the number of points considered in the PCA-AGP and PCA-EGP algorithms by replacing the cell corners with 
the centers of gravity of the cells. 
Afterwords, we expand the resulting box by $\sqrt{3} \epsilon / 2$ to ensure that
the box contains all original points.
We have implemented also these variants, 
but, since for a reasonable big grid size ($\epsilon \geq 0.01$) the running time improvements are negligible, 
we report here only the results of the base variants of the algorithms PCA-AGP and PCA-EGP. However, for very dense grid the improved version of the both algorithms
give better results.

In the following experiments, we add (delete) random points from the point set, and compare the results of a dynamical versions of PCA bounding boxes with the 
their corresponding statical versions (when the covariance matrix of the point set is computed from scratch).
The time of computing, the volume of a bounding box, and the grid density are parameters of interest in this evaluation study.
The test were performed on big number of real graphics models taken from various publicly available sources
(Stanford 3D scanning repository, 3D Cafe). Typical samples of the results are given in 
Table~\ref{table:PCABB-real-lion-add}, Table~\ref{table:PCABB-real-lion-vol}, and Table~\ref{table:PCABB-real-lion-eps-vol}.

\begin{table*}[width=\columnwidth]
\caption{Time needed by the PCA bounding box algorithms for the lion model (183408~points). The values in
the table are the average of results of 100 runs of the algorithms, each time adding/deleting the corresponding number of points.}
\label{table:PCABB-real-lion-add} \centering
\begin{tabular}{|l|c|c|c|c|c|c|}
\hline
\multicolumn{7}{|c|}{Adding/deleting points, $\epsilon =0.005$} \\
\hline
 & 1pnt & 1pnt& 100 pnts & 100 pnts& 1000 pnts& 1000 pnts \\
\hline
 algorithm & static & dynamic& static& dynamic& static& dynamic \\
\hline
PCA-AP      & 0.166 s    & 0.014 s   &  0.171 s &  0.015 s  &  0.172 s  & 0.016 s \\
\hline
PCA-AGP   &  0.092 s &   0.0095 s &   0.093 s&   0.0085 s   &  0.99 s  & 0.017 s\\
\hline
PCA-EGP   &   0.0805 s  & 0.0055 s &  0.082 s &  0.006 s &   0.092 s  &  0.0135 s \\
\hline
\end{tabular}
\end{table*}

The main conclusions of the experiments are as follows:

\begin{itemize}
\item As expected from the theoretical results, the dynamic versions of the algorithms
are significantly faster than
their static counterparts. Typically, the dynamic
versions are about an oder of magnitude faster (see Table~\ref{table:PCABB-real-lion-add}).
\item The dynamic PCA-AP algorithm is not only significantly faster than its statical version, it is
also faster than the static version of the PCA-AGP and PCA-EGP algorithms.
This is due to fact that the brute force manner of finding the extremal points is faster
than computing the covariance matrix of the new point set from scratch, although both
algorithms require $O(n)$ time in the asymptotic analysis.

\item 
Clearly,  the PCA-AGP and PCA-EGP algorithms, that exploit the grid subdivision structure,
are faster than the PCA-AP algorithm. 
The price that must be paid for this is two-folded. 
First, an extra preprocessing time for building the grid is needed.
For the example considered in Table~\ref{table:PCABB-real-lion-add},
computing the grid takes about $0.4$ seconds for the PCA-AGP algorithm, and about $0.43$
for the PCA-EGP algorithm.
Second, the resulting bounding boxes are less precise (see Table~\ref{table:PCABB-real-lion-vol}).

\item  As it is shown in Table~\ref{table:PCABB-real-lion-eps-vol}, for grids that are not very sparse
($\epsilon \leq 0.03$), the approximated PCA bounding boxes computed by 
the PCA-AGP and PCA-EGP algorithms are quite close to the exact PCA bounding boxes.

\end{itemize}

\begin{table*}[width=\columnwidth]
\caption{Volume of the PCA bounding box algorithms for the lion model.
The values in
the table are the average of results of 100 runs of the algorithms, each time adding the corresponding number of points.}
\label{table:PCABB-real-lion-vol} \centering
\begin{tabular}{|l|c|c|c|c|c|}
\hline
\multicolumn{6}{|c|}{Adding points, dynamic version, $\epsilon =0.005$} \\
\hline
 algorithm  & 1pnt & 10pnt& 100 pnts & 1000 pnts& 10000 pnts \\
\hline
PCA-AP      &  285.5  &   644.6 & 856.3 &  1149.1  &   1236.4  \\
\hline
PCA-AGP, PCA-EGP  & 295.5  & 662.7  & 880.3 &  1221.8  &    1263.2  \\
\hline
\end{tabular}
\end{table*}

\begin{table*}[width=\columnwidth]
\caption{Volumes of the PCA bounding boxes algorithms for lion model for different grid density.
The values in
the table are the average of results of 100 runs of the algorithms, each time adding the corresponding number of points.}
\label{table:PCABB-real-lion-eps-vol} \centering
\begin{tabular}{|l|c|c|c|c|c|c|}
\hline
\multicolumn{7}{|c|}{ Adding 100 points, dynamic version} \\
\hline
 algorithm  & $\epsilon =0.005$ & $\epsilon =0.01$ &   $\epsilon =0.03$ & $\epsilon =0.05$ & $\epsilon =0.1$ & $\epsilon =0.2$ \\
\hline
PCA-AP      &   856.3 &  856.3  & 856.3 & 856.3  &  856.3   &  856.3   \\
\hline
PCA-AGP, PCA-EGP  &  880.3 &   904.3  & 942.3 &   1080.1  &  1292.7  &  2324.8   \\
\hline
\end{tabular}
\end{table*}

Tight bounding boxes for the PCA-AGP and PCA-EGP algorithms can be obtained by the following approach. Let $P_1$ be the supporting plane at the extremal grid point along 
one principal direction, and let $P_2$ be the plane parallel to $P_1$, such that the distance between $P_1$ and $P_2$ is
$\sqrt{3} \epsilon / 2$, and $P_2$ intersect or is tangent to the grid. We denote by $S$ the subspace between
$P_1$ and $P_2$. 
Then, the candidates points for the chosen principal direction, that determine the tight bounding box, are all original points
that belong to cells that have intersection with $S$.
See Fig.~\ref{fig:grid-tight-bb} for an illustration. However, in the worst case all original points have to be checked.

\begin{figure}[h!]
  \centering
  \includegraphics[scale=0.51]{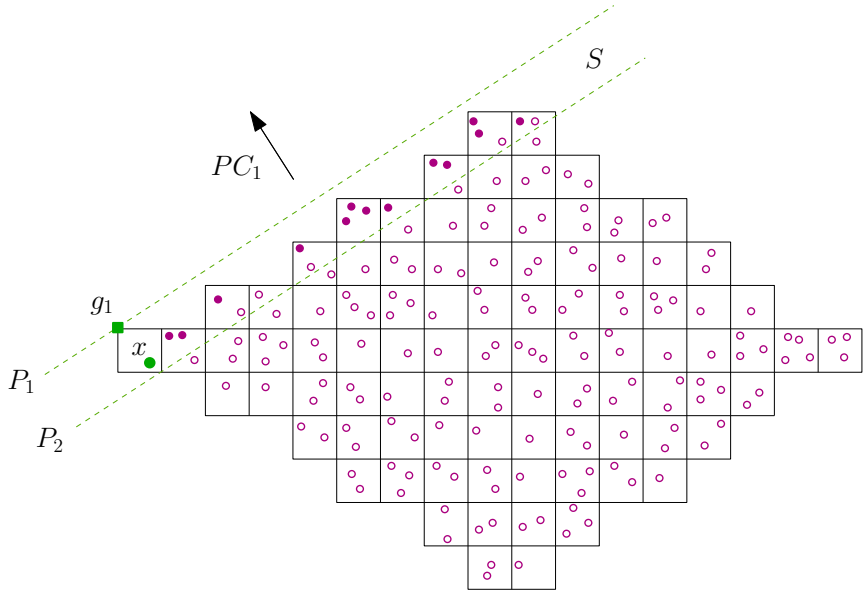}
  \caption{For the principal direction $PC_1$, the algorithms  PCA-AGP and PCA-EGP detect
  the point $g_1$ as extremal grid point, and the point $x$ as extremal point of the original point set.
  However, there are other points (the violet colored circles) that are further than $x$ along $PC_1$.}
  \label{fig:grid-tight-bb}
\end{figure}

Further (theoretical) improvement of the algorithms presented here could be obtained if, instead of the point set, we consider its convex hull 
when we look for extremal points. This only makes sense if the convex hull  is computed
dynamically. Otherwise, computing the static convex hull of the points will be more expensive than finding the exact extremal points
by scanning all points. 

\begin{figure}[h!]
  \centering
  \includegraphics[scale=0.51]{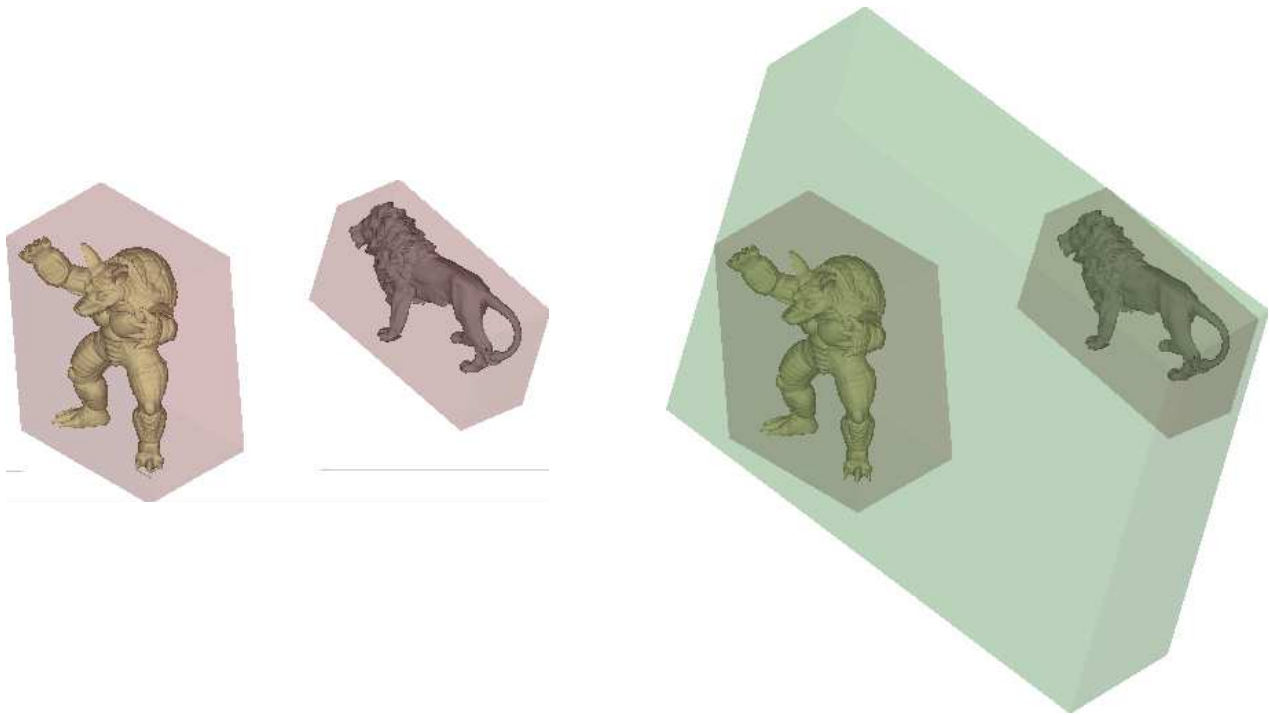}
  \caption{Left: two objects with their PCA bounding boxes. 
           Right: the common PCA bounding box. 
           Computing the common PCA bounding box dynamically takes $0.004$ seconds, while the static version
           takes $0.02$ seconds.}
  \label{fig:common-bb}
\end{figure}
\subsubsection{Computing efficiently a bounding box of several objects}

An interesting application of the closed-form solutions from Section~\ref{sec:PCsDynamicaly1}
is to compute the principal components of two or more objects with already known covariance matrices.
Since $\sigma_{ij}$ and $\sigma_{ij}^m$ in (\ref{eq:add80}) and (\ref{eq:delete80})
are previously known, $\sigma_{ij}'$ can be computed in $O(1)$ time. Thus, for fixed $d$ 
the new covariance matrix $\Sigma$ and the new principal components can be computed also in $O(1)$ time.
This is a significant improvement over the commonly used approach to compute the principal
components from scratch, which take time linear in the number of points.
Efficient computation of the common PCA bounding box of several object 
 is straightforward. See Fig.~\ref{fig:common-bb} for an illustration in $\mathbb{R}^3$.

\section{Conclusion and future work}\label{sec:conclusion}

The main contribution of this paper are
the closed-form solutions for updating the principal components of a dynamic point set.
The new principal components can be computed in constant time,
when a constant
number of points are added or deleted from the point set.
This is a significant improvement of the commonly used approach, when the 
new principal components are computed from scratch, which takes linear time.
The advantages of the theoretical results were verified and
presented in the context of computing dynamic PCA bounding boxes,
a very important application in many fields including computer graphics, where the PCA boxes are used to 
maintain hierarchical data structures for fast rendering of a scene or for collision detection. 
We have presented three practical simple algorithms and compare their performances.

In the appendix we consider the computation of the principal components of a dynamic  continuous point set.
We give closed form-solutions when the point set is a convex polytope or the boundary of a 
convex polytope in $\mathbb{R}^2$ or $\mathbb{R}^3$.

An interesting open problem is to find a closed-form solution for dynamical point sets different
from convex polyhedra, for example, implicit surfaces or B-splines.
An implementation of computing principal components in a dynamic and continuous setting is planned for future work.
Applications of the results presented here in other fields, like computer vision or visualization,
are of high interest.

There are several further improvements and open problems regarding computing dynamic PCA bounding boxes.
Instead of subdividing the space by a simple regular grid, one can use more
sophisticated data structures, like octrees or binary space partition-trees to speed up the time needed to find
the extremal points along the principal directions.
A practical, implementable algorithm for computing the dynamic convex hull of the point set
(computing  extremal point dynamically) would also improve the dynamic PCA bounding box algorithms.
Finding coresets for dynamic PCA bounding boxes will lead to efficient approximation algorithms for PCA bounding boxes.
We are also not aware of data structures for efficient computation of extremal points both approximately 
and dynamically. Such data structures are also of interest.

\section{Appendix \\ \large{Updating the principal components efficiently - continuous case}}


Here, we consider the computation of the principal components of a dynamic  continuous point set.
We present a closed form-solutions when the point set is a convex polytope or a boundary of a 
convex polytope in $\mathbb{R}^2$ or $\mathbb{R}^3$.
When the point set is a boundary of a convex polytope,
we can update the new principal components in $O(k)$ time, 
for both deletion and addition,
under the assumption that we know the $k$ facets in which the polytope changes.
Under the same assumption, 
when the point set is a convex polytope in $\mathbb{R}^2$ or $\mathbb{R}^3$, 
we can update the principal components in $O(k)$ time after adding points.
But, to update the principal components after deleting points from a convex polytope in $\mathbb{R}^2$ or $\mathbb{R}^3$
we need $O(n)$ time.
This is due to the fact that, after a deletion
the center of gravity of the old  convex hull (polyhedron)
could lie outside the new convex hull, and therefore, a retetrahedralization is needed 
(see Subsection~\ref{subsec:cpca-polytope} and Subsection~\ref{subsec:cpca-polygon} for details).

\subsection{Continuous PCA in $\mathbb{R}^3$}

\subsubsection{Continuous PCA over a (convex) polyhedron in $\mathbb{R}^3$}
\label{subsec:cpca-polytope}

Let $P$ be a point set in $\mathbb{R}^3$, and let $X$ be its convex hull.
We assume that the boundary of $X$ is triangulated (if it is not, we can triangulate
it in preprocessing). 
We choose an arbitrary point $\vec{o}$ in the interior of $X$, for example, 
we can
choose that $\vec{o}$ is the center of gravity of the boundary of $X$.
Each triangle from the boundary together with $\vec{o}$ forms a tetrahedron. 
Let the number
of such formed tetrahedra be $n$.
The $k$-th tetrahedron, with vertices 
${\vec{x}_{1,k}}, {\vec{x}_{2,k}}, {\vec{x}_{3,k}}, {\vec{x}_{4,k}} = \vec{o}$, 
can be represented in a parametric form by
$
{\vec{Q}_i}(s,t,u)={\vec{x}_{4,i} +
s \, (\vec{x}_{1,i}} - \vec{x}_{4,i}) +
t \, ( {\vec{x}_{2,i}} - \vec{x}_{4,i} ) +
u \, ( {\vec{x}_{3,i}} - \vec{x}_{4,i} ),
$
for $0 \leq s, t, u \leq 1$,
and $s+t+u \leq 1$.
For $1 \leq i \leq 3$, we use $x_{i,j,k}$ to denote the $i$-th coordinate of the
vertex $\vec{x}_j$ of the polyhedron $\vec{Q}_k$.

The center of gravity of the $k$-th tetrahedron is
$$
\begin{array}{lll}
\vec{\mu}_k& = & \frac{\int_0^1 \int_0^{1-s} \int_0^{1- s -t} \rho({\vec{Q}_k}(s,t)) {\vec{Q}_i}(s,t) \, d u \,d t \, d s}
{\int_0^1 \int_0^{1-s} \int_0^{1- s -t} \rho({\vec{Q}_k}(s,t)) \, d u \, d t \, d s},
\end{array}
$$
where $\rho({\vec{Q}_k}(s,t))$ is a mass density at a point ${\vec{Q}_k}(s,t)$.
Since, we  can assume $\rho({\vec{Q}_k}(s,t)) = 1$, we have
$$
\begin{array}{l}
\vec{\mu}_k \; = \; \frac{\int_0^1 \int_0^{1-s} \int_0^{1- s -t}  {\vec{Q}_k}(s,t) \, d u \,d t \, d s}
{\int_0^1 \int_0^{1-s} \int_0^{1- s -t}  \, d u \, d t \, d s} 
 \; = \; \frac{{\vec{x}_{1,k}}+ {\vec{x}_{2,k}} + {\vec{x}_{3,k}} + \vec{x}_{4,k}}{4}.
\end{array}
$$
The contribution of each tetrahedron to the center of gravity of $X$
is proportional to its volume. If $M_k$ is the $3 \times 3$ matrix whose
$l$-th row is ${\vec{x}_{l,k}}- {\vec{x}_{4,k}}$, for $l=1 \dots 3$, then
the volume of the $k$-th tetrahedron is
$$
v_k=\mbox{volume}(Q_k) =
\frac{
|det(M_k)|}{3!}.
$$
We introduce a weight to each
tetrahedron that is proportional with its volume, define as
$$
w_k = \frac{v_k}{\sum_{k=1}^{n} v_k} = \frac{v_k}{v},
$$
where $v$ is the volume of $X$. Then, the center of gravity of $X$ is
$$
\vec{\mu}  = 
\sum_{k=1}^{n} w_k  \vec{\mu}_k.
$$

The covariance matrix of the $k$-th tetrahedron is
$$
\begin{array}{lll}
\Sigma_k & = &
\frac{\int_0^1 \int_0^{1-s} \int_0^{1- s -t} {(\vec{Q}_k}(s,t,u) - \vec{\mu}) \,
(\vec{Q}_k(s,t,u) - \vec{\mu})^T  \, d u \, d t \, d s}
{\int_0^1 \int_0^{1-s} \int_0^{1- s -t} \, d u \, d t \, d s} \\ 
& = & 
\frac{1}{20} \Big( \sum_{j=1}^{4} \sum_{h=1}^{4}
( {\vec{x}_{j,k}} - \vec{\mu} ) {({\vec{x}_{h,k}} - \vec{\mu} )}^T 
+ \\
& & \quad \;\;\;\, \sum_{j=1}^{4} ( {\vec{x}_{j,k}} - \vec{\mu} ) {({\vec{x}_{j,k}} - \vec{\mu} )}^T  \Big).
\end{array}
$$

The $(i,j)$-th element of $\Sigma_k$, $i, j \in \{1,2,3\}$, is
$$
\begin{array}{lll}
\sigma_{ij,k} & = &
 \frac{1}{20} \Big( \sum_{l=1}^{4} \sum_{h=1}^{4}
( x_{i,l,k} - \mu_i ) (x_{j, h,k} - \mu_j ) 
+ \\
& & \sum_{l=1}^{4} ( x_{i,l,k}- \mu_i ) (x_{j,l,k} - \mu_j ) \Big),
\end{array}
$$
with $\vec{\mu}=(\mu_1, \mu_2, \mu_3 )$.
Finally, the covariance matrix of $X$ is
$$
\begin{array}{lll}
\Sigma & = &
\sum_{i=1}^{n} w_i \Sigma_i,
\end{array}
$$

with $(i,j)$-th element

$$
\begin{array}{lll}
\sigma_{ij} & = &
 \frac{1}{20} \Big( \sum_{k=1}^{n} \sum_{l=1}^{4} \sum_{h=1}^{4}
w_i ( x_{i,l,k} - \mu_i ) (x_{j, h,k} - \mu_j ) 
+ \\
& & \hspace{0.75cm} \sum_{k=1}^{n} \sum_{l=1}^{4} w_i ( x_{i,l,k}- \mu_i ) (x_{j,l,k} - \mu_j ) \Big).
\end{array}
$$

We would like to note that the above expressions hold also
for any non-convex polyhedron that can be tetrahedralized.
A star-shaped object, where $\vec{o}$ is the kernel of the object, is such example.

\newpage
\noindent
{\bf {Adding points}}
\label{subsec:adding-point20}
\bigskip

We add points to $P$, obtaining a new point set $P'$.
Let $X'$ be the convex hull of $P'$.
We consider that $X'$ is  obtained from $X$ by deleting $n_d$, and
adding $n_a$ tetrahedra.
Let
$$
v' = \sum_{k=1}^{n} v_k + \sum_{k=1}^{n_a} v_k -\sum_{k=1}^{n_d} v_k
   = v + \sum_{k=1}^{n_a} v_k -\sum_{k=1}^{n_d} v_k.
$$

The center of gravity of $X'$ is
\begin{equation}\label{eq:add100}
\begin{array}{lll}
\vec{\mu}'  & = & \sum_{k=1}^{n}w_k'\vec{\mu}_k +
            \sum_{k=1}^{n_a}w_k'\vec{\mu}_k -
            \sum_{k=1}^{n_d}w_k'\vec{\mu}_k \vspace{0.3cm} \\
      & = & \frac{1}{v'}\left( \sum_{k=1}^{n}v_k\vec{\mu}_k +
            \sum_{k=1}^{n_a}v_k\vec{\mu}_k -
            \sum_{k=1}^{n_d}v_k\vec{\mu}_k \right) \vspace{0.3cm} \\
      & = & \frac{1}{v'}\left( v \vec{\mu} +
            \sum_{k=1}^{n_a}v_k\vec{\mu}_k -
            \sum_{k=1}^{n_d}v_k\vec{\mu}_k\right).
\end{array}
\end{equation}

Let
$$
\vec{\mu}_a = \frac{1}{v'}\sum_{k=1}^{n_a} v_k\vec{\mu}_k, \;\;\;\text{and}\;\;\;
\vec{\mu}_d = \frac{1}{v'}\sum_{k=1}^{n_d} v_k\vec{\mu}_k.
$$

Then, we can rewrite (\ref{eq:add100}) as
\begin{equation}\label{eq:add110}
\vec{\mu}'  =  \frac{v}{v'} \vec{\mu} + \vec{\mu}_a - \vec{\mu}_d.
\end{equation}

The $i$-th component of $\vec{\mu}_a$ and $\vec{\mu}_d$, $1 \leq i \leq 3$,
is denoted by $\mu_{i,a}$ and $\mu_{i,d}$, respectively.
The $(i,j)$-th component, $\sigma_{ij}'$, $1 \leq i, j \leq 3$, of the covariance matrix 
$\Sigma'$ of $X'$ is
$$
\begin{array}{lll}
\sigma_{ij}' & = &
 \frac{1}{20} \Big( \sum_{k=1}^{n} \sum_{l=1}^{4} \sum_{h=1}^{4}
w_k' ( x_{i,l,k} - \mu_i' ) (x_{j, h,k} - \mu_j' ) + \vspace{0.3cm}\\
& & 
\hspace{0.75cm} \sum_{k=1}^{n} \sum_{l=1}^{4} w_k' ( x_{i,l,k}- \mu_i' ) (x_{j,l,k} - \mu_j ') \Big) + \vspace{0.3cm}\\
& &
\frac{1}{20} \Big( \sum_{k=1}^{n_a} \sum_{l=1}^{4} \sum_{h=1}^{4}
w_k' ( x_{i,l,k} - \mu_i' ) (x_{j, h,k} - \mu_j' ) + \vspace{0.3cm}\\
& & 
\hspace{0.75cm} \sum_{k=1}^{n_a} \sum_{l=1}^{4} w_k' ( x_{i,l,k}- \mu_i' ) (x_{j,l,k} - \mu_j ') - \vspace{0.3cm}\\
& &
\hspace{0.75cm}\sum_{k=1}^{n_d} \sum_{l=1}^{4} \sum_{h=1}^{4}
w_k' ( x_{i,l,k} - \mu_i' ) (x_{j, h,k} - \mu_j' ) - \vspace{0.3cm}\\
& & 
\hspace{0.75cm} \sum_{k=1}^{n_d} \sum_{l=1}^{4} w_k' ( x_{i,l,k}- \mu_i' ) (x_{j,l,k} - \mu_j ')\Big).
\end{array}
$$

Let

$$
\sigma_{ij}' = \frac{1}{20} (\sigma_{ij,11}' + \sigma_{ij,12}' + \sigma_{ij,21}'+ \sigma_{ij,22}'
               - \sigma_{ij,31}'- \sigma_{ij,32}'),
$$

where,
\begin{equation}\label{eq:add130}
\sigma_{ij,11}' =  \sum_{k=1}^{n} \sum_{l=1}^{4} \sum_{h=1}^{4}
w_k' ( x_{i,l,k} - \mu_i' ) (x_{j, h,k} - \mu_j' ), 
\end{equation}

\begin{equation}\label{eq:add140}
\sigma_{ij,12}' =  \sum_{k=1}^{n} \sum_{l=1}^{4} w_k' ( x_{i,l,k}- \mu_i' ) (x_{j,l,k} - \mu_j '),
\end{equation}

\begin{equation}\label{eq:add150}
\sigma_{ij,21}' =  \sum_{k=1}^{n_a} \sum_{l=1}^{4} \sum_{h=1}^{4}
w_k' ( x_{i,l,k} - \mu_i' ) (x_{j, h,k} - \mu_j' ), 
\end{equation}

\begin{equation}\label{eq:add160}
\sigma_{ij,22}' = \sum_{k=1}^{n_a} \sum_{l=1}^{4} w_k' ( x_{i,l,k}- \mu_i' ) (x_{j,l,k} - \mu_j '),
\end{equation}

\begin{equation}\label{eq:add170}
\sigma_{ij,31}' =  \sum_{k=1}^{n_d} \sum_{l=1}^{4} \sum_{h=1}^{4}
w_k' ( x_{i,l,k} - \mu_i' ) (x_{j, h,k} - \mu_j' ), 
\end{equation}

\begin{equation}\label{eq:add180}
\sigma_{ij,32}' = \sum_{k=1}^{n_d} \sum_{l=1}^{4} w_k' ( x_{i,l,k}- \mu_i' ) (x_{j,l,k} - \mu_j ').
\end{equation}

Plugging-in the values of $\mu_i'$ and $\mu_j'$ in (\ref{eq:add130}), we obtain:

\begin{equation}\label{eq:add190}
\begin{array}{lll}

\sigma_{ij,11}'& = & \sum_{k=1}^{n} \sum_{l=1}^{4} \sum_{h=1}^{4}
w_k' ( x_{i,l,k} - \frac{v}{v'} \mu_i - \mu_{i,a} + \mu_{i,d} )
     (x_{j, h,k} - \frac{v}{v'} \mu_j - \mu_{j,a} + \mu_{j,d} ) \vspace{0.3cm} \\
              & = & \sum_{k=1}^{n} \sum_{l=1}^{4} \sum_{h=1}^{4}
                    w_k' ( x_{i,l,k} - \mu_i + \mu_i(1- \frac{v}{v'}) - \mu_{i,a} + \mu_{i,d} ) \vspace{0.3cm} \\
              & &\hspace{3.7cm}(x_{j, h,k} - \mu_j + \mu_j(1-\frac{v}{v'})  - \mu_{j,a} + \mu_{j,d} ) \vspace{0.3cm} \\
              & = & \sum_{k=1}^{n} \sum_{l=1}^{4} \sum_{h=1}^{4}
                      w_k' ( x_{i,l,k} - \mu_i)(x_{j, h,k} - \mu_j) + \vspace{0.3cm} \\
              & &  \sum_{k=1}^{n} \sum_{l=1}^{4} \sum_{h=1}^{4}
                      w_k' ( x_{i,l,k} - \mu_i)(\mu_j(1-\frac{v}{v'})  - \mu_{j,a} + \mu_{j,d} ) + \vspace{0.3cm} \\
              & & \sum_{k=1}^{n} \sum_{l=1}^{4} \sum_{h=1}^{4}
                      w_k' (\mu_i(1- \frac{v}{v'}) - \mu_{i,a} + \mu_{i,d})(x_{j, h,k} - \mu_j) + \vspace{0.3cm} \\
              & & \sum_{k=1}^{n} \sum_{l=1}^{4} \sum_{h=1}^{4}
                      w_k' (\mu_i(1- \frac{v}{v'}) - \mu_{i,a} + \mu_{i,d})(\mu_j(1-\frac{v}{v'})  - \mu_{j,a} + \mu_{j,d}).
\end{array}
\end{equation}

Since $ \sum_{k=1}^{n} \sum_{l=1}^{4} w_k' ( x_{i,l,k} - \mu_i)=0$, $1 \leq i \leq 3$, we have

\begin{equation}\label{eq:add200}
\begin{array}{lll}
\sigma_{ij,11}' & = & \frac{1}{v'}\sum_{k=1}^{n} \sum_{l=1}^{4} \sum_{h=1}^{4}
                      v_k ( x_{i,l,k} - \mu_i)(x_{j, h,k} - \mu_j) + \vspace{0.3cm} \\
                &   &  \frac{1}{v'} \sum_{k=1}^{n} \sum_{l=1}^{4} \sum_{h=1}^{4}
      v_k (\mu_i(1- \frac{v}{v'}) - \mu_{i,a} + \mu_{i,d})(\mu_j(1-\frac{v}{v'})  - \mu_{j,a} + \mu_{j,d})\vspace{0.3cm} \\
      & = & \frac{1}{v'}\sum_{k=1}^{n} \sum_{l=1}^{4} \sum_{h=1}^{4}
                      v_k ( x_{i,l,k} - \mu_i)(x_{j, h,k} - \mu_j) + \vspace{0.3cm} \\
                &   &  16 \frac{v}{v'}(\mu_i(1- \frac{v}{v'}) - \mu_{i,a} + \mu_{i,d})(\mu_j(1-\frac{v}{v'})  - \mu_{j,a} + \mu_{j,d}).
\end{array}
\end{equation}

Plugging-in the values of $\mu_i'$ and $\mu_j'$ in (\ref{eq:add140}), we obtain:

\begin{equation}\label{eq:add210}
\begin{array}{lll}
\sigma_{ij,12}' & = & \sum_{k=1}^{n} \sum_{l=1}^{4}
w_k' ( x_{i,l,k} - \frac{v}{v'} \mu_i - \mu_{i,a} + \mu_{i,d} )
     (x_{j, h,k} - \frac{v}{v'} \mu_j - \mu_{j,a} + \mu_{j,d} ) \vspace{0.3cm} \\
              & = & \sum_{k=1}^{n} \sum_{l=1}^{4}
                    w_k' ( x_{i,l,k} - \mu_i + \mu_i(1- \frac{v}{v'}) - \mu_{i,a} + \mu_{i,d} ) \vspace{0.3cm} \\
              & &\hspace{2.6cm}(x_{j, h,k} - \mu_j + \mu_j(1-\frac{v}{v'})  - \mu_{j,a} + \mu_{j,d} ) \vspace{0.3cm} \\
              & = & \sum_{k=1}^{n} \sum_{l=1}^{4} 
                      w_k' ( x_{i,l,k} - \mu_i)(x_{j, h,k} - \mu_j) + \vspace{0.3cm} \\
              & &  \sum_{k=1}^{n} \sum_{l=1}^{4} 
                      w_k' ( x_{i,l,k} - \mu_i)(\mu_j(1-\frac{v}{v'})  - \mu_{j,a} + \mu_{j,d} ) + \vspace{0.3cm} \\
              & & \sum_{k=1}^{n} \sum_{l=1}^{4} 
                      w_k' (\mu_i(1- \frac{v}{v'}) - \mu_{i,a} + \mu_{i,d})(x_{j, h,k} - \mu_j) + \vspace{0.3cm} \\
              & & \sum_{k=1}^{n} \sum_{l=1}^{4} 
                      w_k' (\mu_i(1- \frac{v}{v'}) - \mu_{i,a} + \mu_{i,d})(\mu_j(1-\frac{v}{v'})  - \mu_{j,a} + \mu_{j,d}).
\end{array}
\end{equation}

Since $ \sum_{k=1}^{n} \sum_{l=1}^{4} w_k' ( x_{i,l,k} - \mu_i)=0$, $1 \leq i \leq 3$, we have

\begin{equation}\label{eq:add220}
\begin{array}{lll}
\sigma_{ij,12}' & = & \frac{1}{v'}\sum_{k=1}^{n} \sum_{l=1}^{4}
                      v_k ( x_{i,l,k} - \mu_i)(x_{j, h,k} - \mu_j) + \vspace{0.3cm} \\
                &   &  \frac{1}{v'} \sum_{k=1}^{n} \sum_{l=1}^{4}
      v_k (\mu_i(1- \frac{v}{v'}) - \mu_{i,a} + \mu_{i,d})(\mu_j(1-\frac{v}{v'})  - \mu_{j,a} + \mu_{j,d})\vspace{0.3cm} \\
      & = & \frac{1}{v'}\sum_{k=1}^{n} \sum_{l=1}^{4}
                      v_k ( x_{i,l,k} - \mu_i)(x_{j, h,k} - \mu_j) + \vspace{0.3cm} \\
                &   &  4\frac{v}{v'}(\mu_i(1- \frac{v}{v'}) - \mu_{i,a} + \mu_{i,d})(\mu_j(1-\frac{v}{v'})  - \mu_{j,a} + \mu_{j,d}).
\end{array}
\end{equation}

From (\ref{eq:add210}) and (\ref{eq:add220}), we obtain

\begin{equation}\label{eq:add230}
\begin{array}{lll}
\sigma_{ij,1}' & = & \sigma_{ij,11}' + \sigma_{ij,12}' \vspace{0.3cm}\\
               & = & \sigma_{ij} + 20 \frac{v}{v'}(\mu_i(1- \frac{v}{v'}) - \mu_{i,a} + \mu_{i,d})(\mu_j(1-\frac{v}{v'})  - \mu_{j,a} + \mu_{j,d}).
      \end{array}
\end{equation}

Note that $\sigma_{ij,1}'$ can be computed in $O(1)$ time.
The components $\sigma_{ij,21}'$ and $\sigma_{ij,22}'$ can be computed in $O(n_a)$ time,
while $O(n_d)$ time is needed for computing $\sigma_{ij,31}'$ and  $\sigma_{ij,32}'$.
Thus, $\vec{\mu'}$ and

\begin{equation}\label{eq:add240}
\begin{array}{lll}
\sigma_{ij}' & = & \frac{1}{20}(\sigma_{ij,11}' + \sigma_{ij,12}' + \sigma_{ij,21}' + \sigma_{ij,22}' +
                   \sigma_{ij,31}' + \sigma_{ij,32}') \vspace{0.3cm}\\
               & = & \frac{1}{20}(\sigma_{ij} + \sigma_{ij,21}' + \sigma_{ij,22}' + \sigma_{ij,31}' + \sigma_{ij,32}) +  \vspace{0.3cm}\\ 
               & & \frac{v}{v'}(\mu_i(1- \frac{v}{v'}) - \mu_{i,a} + \mu_{i,d})(\mu_j(1-\frac{v}{v'})  - \mu_{j,a} + \mu_{j,d})
\end{array}
\end{equation}

can be computed in $O(n_a + n_d)$ time.

\bigskip
\noindent
{\bf {Deleting points}}
\label{subsec:deleting-point20}
\bigskip
 
Let the new convex hull be obtained by deleting
$n_d$ tetrahedra from and added $n_a$ tetrahedra to the old convex hull. 
If the interior point $\vec{o}$ (needed for a tetrahedronization of a convex 
polytope), after deleting points, lies inside the new convex hull, then
the same formulas and time complexity, as by adding points, follow.
If $\vec{o}$ lie outside the new convex hull, then, we need
to choose a new interior point $\vec{o}'$, and recompute the new tetrahedra associated with it.
Thus, we need in total $O(n)$ time to update the principal components.

Under certain assumptions, we can recompute the new principal components faster:

\begin{itemize}

\item If we know that a certain point of the polyhedron will never be deleted, we can choose $\vec{o}$
to be that point. In that case, we also have the same 
closed-formed solution as for adding a point.

\item Let the facets of the convex polyhedron have similar (uniformly distributed) area.
We choose $\vec{o}$ to be the center of gravity of the polyhedron. Then, we can expect
that after deleting a point, $\vec{o}$ will remain in the new convex hull.
However, after several deletion, $\vec{o}$ could lie outside the convex hull, and then
we need to recompute it and the associate tetrahedra with it.

\end{itemize}

Note, that in the case when we consider boundary of a convex polyhedron 
(Subsection~\ref{subsec:cpca-polytope-boundary} and Subsection~\ref{subsec:cpca-polygon-boundary}),
we do not need an interior point $\vec{o}$
and the same time complexity holds for both adding and deleting points.

\subsubsection{Continuous PCA over a boundary of a polyhedron}
\label{subsec:cpca-polytope-boundary}

Let $X$ be a polyhedron in $\mathbb{R}^3$.
We assume that the boundary of  $X$ is triangulated (if it is not, we can 
triangulate it in preprocessing), containing
$n$ triangles.
The $k$-th triangle, with vertices 
${\vec{x}_{1,k}}, {\vec{x}_{2,k}}, {\vec{x}_{3,k}}$, can be represented 
in a parametric form by
$
{\vec{T}_k}(s,t)={\vec{x}_{1,k}} + 
s \, ( {\vec{x}_{2,k}} - {\vec{x}_{1,k}} ) +
t \, ( {\vec{x}_{3,k}} - {\vec{x}_{1,k}} ),
$
for $0 \leq s, t \leq 1$,
and $s+t \leq 1$. 
For $1 \leq i \leq 3$, we denote by $x_{i, j, k}$ the $i$-th coordinate of the vertex $\vec{x}_j$
of the triangle $\vec{T}_k$.

The center of gravity of the $k$-th triangle is
$$
\vec{\mu}_k=\frac{\int_0^1 \int_0^{1-s} {\vec{T}_i}(s,t) \, d t \, d s}
{\int_0^1 \int_0^{1-s} d t \, d s} =
\frac{{\vec{x}_{1,k}}+ {\vec{x}_{2,k}} + {\vec{x}_{3,k}}}{3}.
$$
The contribution of each triangle to the center of gravity of the triangulated
surface is proportional to its area.
The area of the $k$-th triangle is
$$
a_k=\mbox{area}(T_k) =
\frac{
|( {\vec{x}_{2,k}} - {\vec{x}_{1,k}} )| \times
|( {\vec{x}_{3,k}} - {\vec{x}_{1,k}} )|}{2}.
$$
We introduce a weight to each
triangle that is proportional with its area, define as
$$
w_k = \frac{a_k}{\sum_{i=1}^{n} a_k} = \frac{a_k}{a},
$$
where $a$ is the area of $X$. Then, the center of gravity of the boundary of $X$ is
$$
\vec{\mu}  = 
\sum_{k=1}^{n} w_k \vec{\mu}_k.
$$
The covariance matrix of the $k$-th triangle is
$$
\begin{array}{lll}
\Sigma_k & = &
\frac{\int_0^1 \int_0^{1-s} {(\vec{T}_k}(s,t) - \vec{\mu}) \,
(\vec{T}_k (s,t) - \vec{\mu})^T \, d t \, d s}
{\int_0^1 \int_0^{1-s} \, d t \, d s} \\
& = & 
\frac{1}{12} \Big( \sum_{j=1}^{3} \sum_{h=1}^{3}
( {\vec{x}_{j,k}} - \vec{\mu} ) {({\vec{x}_{h,k}} - \vec{\mu} )}^T 
+ \\
& & \quad \;\;\;\, \sum_{j=1}^{3} ( {\vec{x}_{j,k}} - \vec{\mu} ) {({\vec{x}_{j,k}} - \vec{\mu} )}^T \Big).
\end{array}
$$
The $(i,j)$-th element of $\Sigma_k$, $i, j\in \{1, 2, 3\}$, is 
$$
\begin{array}{lll}
\sigma_{ij, k} & = &
 \frac{1}{12} \Big( \sum_{l=1}^{3} \sum_{h=1}^{3}
( x_{i, l, k} - \mu_i ) (x_{j, h, k} - \mu_j ) 
+ \\
& & \quad \;\;\;\, \sum_{l=1}^{3} ( x_{i, l, k} - \mu_i ) (x_{j, l, k} - \mu_j ) \Big),
\end{array}
$$
with $\vec{\mu}=(\mu_1, \mu_2, \mu_3 )$.
Finally, the covariance matrix of the boundary of $X$ is
$$
\begin{array}{lll}
\Sigma & = &
\sum_{k=1}^{n} w_k \Sigma_k.
\end{array}
$$

\bigskip
\noindent
{\bf {Adding points}}
\label{subsec:adding-point20-B3D}
\bigskip

We add points to $X$.
Let $X'$ be the new convex hull.
We assume that $X'$ is  obtained from $X$ by deleting $n_d$, and
adding $n_a$ tetrahedra.
Then the sum of the areas of all triangles is
$$
a' = \sum_{k=1}^{n} a_k + \sum_{k=1}^{n_a} a_k -\sum_{k=1}^{n_d} a_k
   = a + \sum_{k=1}^{n_a} a_k -\sum_{k=1}^{n_d} a_k.
$$

The center of gravity of $X'$ is
\begin{equation}\label{eq:add100-B3D}
\begin{array}{lll}
\vec{\mu}'  & = & \sum_{k=1}^{n}w_k'\vec{\mu}_k +
            \sum_{k=1}^{n_a}w_k'\vec{\mu}_k -
            \sum_{k=1}^{n_d}w_k'\vec{\mu}_k \vspace{0.3cm} \\
      & = & \frac{1}{a'}\left( \sum_{k=1}^{n}a_k\vec{\mu}_k +
            \sum_{k=1}^{n_a}a_k\vec{\mu}_k -
            \sum_{k=1}^{n_d}a_k\vec{\mu}_k \right) \vspace{0.3cm} \\
      & = & \frac{1}{a'}\left( a \vec{\mu} +
            \sum_{k=1}^{n_a}a_k\vec{\mu}_k -
            \sum_{k=1}^{n_d}a_k\vec{\mu}_k\right).
\end{array}
\end{equation}

Let
$$
\vec{\mu}_a = \frac{1}{a'}\sum_{k=1}^{n_a} a_k\vec{\mu}_k, \;\;\;\text{and}\;\;\;
\vec{\mu}_d = \frac{1}{a'}\sum_{k=1}^{n_d} a_k\vec{\mu}_k.
$$

Then, we can rewrite (\ref{eq:add100-B3D}) as
\begin{equation}\label{eq:add110-B3D}
\vec{\mu}'  =  \frac{a}{a'} \vec{\mu} + \vec{\mu}_a - \vec{\mu}_d.
\end{equation}

The $i$-th component of $\vec{\mu}_a$ and $\vec{\mu}_d$, $1 \leq i \leq 3$,
is denoted by $\mu_{i,a}$ and $\mu_{i,d}$, respectively.
The $(i,j)$-th component, $\sigma_{ij}'$, $1 \leq i, j \leq 3$, of the covariance matrix 
$\Sigma'$ of $X'$ is
%
$$
\begin{array}{lll}
\sigma_{ij}' & = &
 \frac{1}{12} \Big( \sum_{k=1}^{n} \sum_{l=1}^{3} \sum_{h=1}^{3}
w_k' ( x_{i,l,k} - \mu_i' ) (x_{j, h,k} - \mu_j' ) + \vspace{0.3cm}\\
& & 
\hspace{0.75cm} \sum_{k=1}^{n} \sum_{l=1}^{3} w_k' ( x_{i,l,k}- \mu_i' ) (x_{j,l,k} - \mu_j ') \Big) + \vspace{0.3cm}\\
& &
\frac{1}{12} \Big( \sum_{k=1}^{n_a} \sum_{l=1}^{3} \sum_{h=1}^{3}
w_k' ( x_{i,l,k} - \mu_i' ) (x_{j, h,k} - \mu_j' ) + \vspace{0.3cm}\\
& & 
\hspace{0.75cm} \sum_{k=1}^{n_a} \sum_{l=1}^{3} w_k' ( x_{i,l,k}- \mu_i' ) (x_{j,l,k} - \mu_j ') - \vspace{0.3cm}\\
& &
\hspace{0.75cm}\sum_{k=1}^{n_d} \sum_{l=1}^{3} \sum_{h=1}^{3}
w_k' ( x_{i,l,k} - \mu_i' ) (x_{j, h,k} - \mu_j' ) - \vspace{0.3cm}\\
& & 
\hspace{0.75cm} \sum_{k=1}^{n_d} \sum_{l=1}^{3} w_k' ( x_{i,l,k}- \mu_i' ) (x_{j,l,k} - \mu_j ')\Big).
\end{array}
$$

Let

$$
\sigma_{ij}' = \frac{1}{12} (\sigma_{ij,11}' + \sigma_{ij,12}' + \sigma_{ij,21}'+ \sigma_{ij,22}'
               - \sigma_{ij,31}'- \sigma_{ij,32}'),
$$

where,
\begin{equation}\label{eq:add130-B3D}
\sigma_{ij,11}' =  \sum_{k=1}^{n} \sum_{l=1}^{3} \sum_{h=1}^{3}
w_k' ( x_{i,l,k} - \mu_i' ) (x_{j, h,k} - \mu_j' ), 
\end{equation}

\begin{equation}\label{eq:add140-B3D}
\sigma_{ij,12}' =  \sum_{k=1}^{n} \sum_{l=1}^{3} w_k' ( x_{i,l,k}- \mu_i' ) (x_{j,l,k} - \mu_j '),
\end{equation}

\begin{equation}\label{eq:add150-B3D}
\sigma_{ij,21}' =  \sum_{k=1}^{n_a} \sum_{l=1}^{3} \sum_{h=1}^{3}
w_k' ( x_{i,l,k} - \mu_i' ) (x_{j, h,k} - \mu_j' ), 
\end{equation}

\begin{equation}\label{eq:add160-B3D}
\sigma_{ij,22}' = \sum_{k=1}^{n_a} \sum_{l=1}^{3} w_k' ( x_{i,l,k}- \mu_i' ) (x_{j,l,k} - \mu_j '),
\end{equation}

\begin{equation}\label{eq:add170-B3D}
\sigma_{ij,31}' =  \sum_{k=1}^{n_d} \sum_{l=1}^{3} \sum_{h=1}^{3}
w_k' ( x_{i,l,k} - \mu_i' ) (x_{j, h,k} - \mu_j' ), 
\end{equation}

\begin{equation}\label{eq:add180-B3D}
\sigma_{ij,32}' = \sum_{k=1}^{n_d} \sum_{l=1}^{3} w_k' ( x_{i,l,k}- \mu_i' ) (x_{j,l,k} - \mu_j ').
\end{equation}

Plugging-in the values of $\mu_i'$ and $\mu_j'$ in (\ref{eq:add130-B3D}), we obtain:

\begin{equation}\label{eq:add190-B3D}
\begin{array}{lll}

\sigma_{ij,11}'& = & \sum_{k=1}^{n} \sum_{l=1}^{3} \sum_{h=1}^{3}
w_k' ( x_{i,l,k} - \frac{a}{a'} \mu_i - \mu_{i,a} + \mu_{i,d} )
     (x_{j, h,k} - \frac{a}{a'} \mu_j - \mu_{j,a} + \mu_{j,d} ) \vspace{0.3cm} \\
              & = & \sum_{k=1}^{n} \sum_{l=1}^{3} \sum_{h=1}^{3}
                    w_k' ( x_{i,l,k} - \mu_i + \mu_i(1- \frac{a}{a'}) - \mu_{i,a} + \mu_{i,d} ) \vspace{0.3cm} \\
              & &\hspace{3.7cm}(x_{j, h,k} - \mu_j + \mu_j(1-\frac{a}{a'})  - \mu_{j,a} + \mu_{j,d} ) \vspace{0.3cm} \\
              & = & \sum_{k=1}^{n} \sum_{l=1}^{3} \sum_{h=1}^{3}
                      w_k' ( x_{i,l,k} - \mu_i)(x_{j, h,k} - \mu_j) + \vspace{0.3cm} \\
              & &  \sum_{k=1}^{n} \sum_{l=1}^{3} \sum_{h=1}^{3}
                      w_k' ( x_{i,l,k} - \mu_i)(\mu_j(1-\frac{a}{a'})  - \mu_{j,a} + \mu_{j,d} ) + \vspace{0.3cm} \\
              & & \sum_{k=1}^{n} \sum_{l=1}^{3} \sum_{h=1}^{3}
                      w_k' (\mu_i(1- \frac{a}{a'}) - \mu_{i,a} + \mu_{i,d})(x_{j, h,k} - \mu_j) + \vspace{0.3cm} \\
              & & \sum_{k=1}^{n} \sum_{l=1}^{3} \sum_{h=1}^{3}
                      w_k' (\mu_i(1- \frac{a}{a'}) - \mu_{i,a} + \mu_{i,d})(\mu_j(1-\frac{a}{a'})  - \mu_{j,a} + \mu_{j,d}).
\end{array}
\end{equation}

Since $ \sum_{k=1}^{n} \sum_{l=1}^{3} w_k' ( x_{i,l,k} - \mu_i)=0$, $1 \leq i \leq 3$, we have

\begin{equation}\label{eq:add200-B3D}
\begin{array}{lll}
\sigma_{ij,11}' & = & \frac{1}{a'}\sum_{k=1}^{n} \sum_{l=1}^{3} \sum_{h=1}^{3}
                      a_k ( x_{i,l,k} - \mu_i)(x_{j, h,k} - \mu_j) + \vspace{0.3cm} \\
                &   &  \frac{1}{a'} \sum_{k=1}^{n} \sum_{l=1}^{3} \sum_{h=1}^{3}
      a_k (\mu_i(1- \frac{a}{a'}) - \mu_{i,a} + \mu_{i,d})(\mu_j(1-\frac{a}{a'})  - \mu_{j,a} + \mu_{j,d})\vspace{0.3cm} \\
      & = & \frac{1}{a'}\sum_{k=1}^{n} \sum_{l=1}^{3} \sum_{h=1}^{3}
                      a_k ( x_{i,l,k} - \mu_i)(x_{j, h,k} - \mu_j) + \vspace{0.3cm} \\
                &   &  9 \frac{a}{a'}(\mu_i(1- \frac{a}{a'}) - \mu_{i,a} + \mu_{i,d})(\mu_j(1-\frac{a}{a'})  - \mu_{j,a} + \mu_{j,d}).
\end{array}
\end{equation}

Plugging-in the values of $\mu_i'$ and $\mu_j'$ in (\ref{eq:add140-B3D}), we obtain:

\begin{equation}\label{eq:add210-B3D}
\begin{array}{lll}
\sigma_{ij,12}' & = & \sum_{k=1}^{n} \sum_{l=1}^{3}
w_k' ( x_{i,l,k} - \frac{a}{a'} \mu_i - \mu_{i,a} + \mu_{i,d} )
     (x_{j, h,k} - \frac{a}{a'} \mu_j - \mu_{j,a} + \mu_{j,d} ) \vspace{0.3cm} \\
              & = & \sum_{k=1}^{n} \sum_{l=1}^{3}
                    w_k' ( x_{i,l,k} - \mu_i + \mu_i(1- \frac{a}{a'}) - \mu_{i,a} + \mu_{i,d} ) \vspace{0.3cm} \\
              & &\hspace{2.6cm}(x_{j, h,k} - \mu_j + \mu_j(1-\frac{a}{a'})  - \mu_{j,a} + \mu_{j,d} ) \vspace{0.3cm} \\
              & = & \sum_{k=1}^{n} \sum_{l=1}^{3} 
                      w_k' ( x_{i,l,k} - \mu_i)(x_{j, h,k} - \mu_j) + \vspace{0.3cm} \\
              & &  \sum_{k=1}^{n} \sum_{l=1}^{3} 
                      w_k' ( x_{i,l,k} - \mu_i)(\mu_j(1-\frac{a}{a'})  - \mu_{j,a} + \mu_{j,d} ) + \vspace{0.3cm} \\
              & & \sum_{k=1}^{n} \sum_{l=1}^{3} 
                      w_k' (\mu_i(1- \frac{a}{a'}) - \mu_{i,a} + \mu_{i,d})(x_{j, h,k} - \mu_j) + \vspace{0.3cm} \\
              & & \sum_{k=1}^{n} \sum_{l=1}^{3} 
                      w_k' (\mu_i(1- \frac{a}{a'}) - \mu_{i,a} + \mu_{i,d})(\mu_j(1-\frac{a}{a'})  - \mu_{j,a} + \mu_{j,d}).
\end{array}
\end{equation}

Since $ \sum_{k=1}^{n} \sum_{l=1}^{3} w_k' ( x_{i,l,k} - \mu_i)=0$, $1 \leq i \leq 3$, we have

\begin{equation}\label{eq:add220-B3D}
\begin{array}{lll}
\sigma_{ij,12}' & = & \frac{1}{a'}\sum_{k=1}^{n} \sum_{l=1}^{3}
                      a_k ( x_{i,l,k} - \mu_i)(x_{j, h,k} - \mu_j) + \vspace{0.3cm} \\
                &   &  \frac{1}{a'} \sum_{k=1}^{n} \sum_{l=1}^{3}
      a_k (\mu_i(1- \frac{a}{a'}) - \mu_{i,a} + \mu_{i,d})(\mu_j(1-\frac{a}{a'})  - \mu_{j,a} + \mu_{j,d})\vspace{0.3cm} \\
      & = & \frac{1}{a'}\sum_{k=1}^{n} \sum_{l=1}^{3}
                      a_k ( x_{i,l,k} - \mu_i)(x_{j, h,k} - \mu_j) + \vspace{0.3cm} \\
                &   &  3 \frac{a}{a'}(\mu_i(1- \frac{a}{a'}) - \mu_{i,a} + \mu_{i,d})(\mu_j(1-\frac{a}{a'})  - \mu_{j,a} + \mu_{j,d}).
\end{array}
\end{equation}

From (\ref{eq:add210-B3D}) and (\ref{eq:add220-B3D}), we obtain

\begin{equation}\label{eq:add230-B3D}
\begin{array}{lll}
\sigma_{ij,1}' & = & \sigma_{ij,11}' + \sigma_{ij,12}' \vspace{0.3cm}\\
               & = & \sigma_{ij} + 12 \frac{a}{a'}(\mu_i(1- \frac{a}{a'}) - \mu_{i,a} + \mu_{i,d})(\mu_j(1-\frac{a}{a'})  - \mu_{j,a} + \mu_{j,d}).
      \end{array}
\end{equation}

Note that $\sigma_{ij,1}'$ can be computed in $O(1)$ time.
The components $\sigma_{ij,21}'$ and $\sigma_{ij,22}'$ can be computed in $O(n_a)$ time,
while $O(n_d)$ time is needed for computing $\sigma_{ij,31}'$ and  $\sigma_{ij,32}'$.
Thus, $\vec{\mu'}$ and

\begin{equation}\label{eq:add240-B3D}
\begin{array}{lll}
\sigma_{ij}' & = & \frac{1}{12}(\sigma_{ij,11}' + \sigma_{ij,12}' + \sigma_{ij,21}' + \sigma_{ij,22}' +
                   \sigma_{ij,31}' + \sigma_{ij,32}') \vspace{0.3cm}\\
               & = & \frac{1}{12}(\sigma_{ij} + \sigma_{ij,21}' + \sigma_{ij,22}' + \sigma_{ij,31}' + \sigma_{ij,32}) +  \vspace{0.3cm}\\ 
               & & \frac{a}{a'}(\mu_i(1- \frac{a}{a'}) - \mu_{i,a} + \mu_{i,d})(\mu_j(1-\frac{a}{a'})  - \mu_{j,a} + \mu_{j,d}).
\end{array}
\end{equation}
can be computed in $O(n_a + n_d)$ time.

\bigskip
\noindent
{\bf {Deleting points}}
\label{subsec:deleting-point20-B3D}
\bigskip

Let the new convex hull be obtained by deleting
$n_d$ tetrahedra from and added $n_a$ tetrahedra to the old convex hull. 
Consequently,
the same formulas and time complexity, as by adding points, follow.

\subsection{Continuous PCA in $\mathbb{R}^2$}
\label{subsec:cpca-R3-2}

\subsubsection{Continuous PCA over a polygon}
\label{subsec:cpca-polygon}

We assume that the polygon $X$ is triangulated (if it is not, we can triangulate it in preprocessing), and the number of triangles is $n$.
The $k$-th triangle, with vertices 
${\vec{x}_{1,k}}, {\vec{x}_{2,k}}, {\vec{x}_{3,k}} = \vec{o}$, 
can be represented in a parametric form by
$
{\vec{T}_i}(s,t)={\vec{x}_{3,k} +
s \, (\vec{x}_{1,k}} - \vec{x}_{3,k}) +
t \, ( {\vec{x}_{2,k}} - \vec{x}_{3,k} ),
$ 
for $\;$ $0 \leq s, t \leq 1$,
and $s+t \leq 1$.

The center of gravity of the $k$-th triangle is
$$
\vec{\mu}_i=\frac{\int_0^1 \int_0^{1-s} {\vec{T}_i}(s,t) \,d t \, d s}
{\int_0^1 \int_0^{1-s}  \, d t \, d s} =
\frac{{\vec{x}_{1,k}}+ {\vec{x}_{2,k}} + {\vec{x}_{3,k}}}{3}.
$$
The contribution of each triangle to the center of gravity of $X$
is proportional to its area. 
The area of the $i$-th triangle is
$$
a_k=\mbox{area}(T_k) =
\frac{
|( {\vec{x}_{2,k}} - {\vec{x}_{1,k}} )| \times
|( {\vec{x}_{3,k}} - {\vec{x}_{1,k}} )|}{2},
$$
where $\times$ denotes the vector product. We introduce a weight to each
triangle that is proportional with its area, define as
$$
w_k = \frac{a_k}{\sum_{k=1}^{n} a_k} = \frac{a_k}{a},
$$
where $a$ is the area of $X$.Then, the center of gravity of $X$ is
$$
\vec{\mu}  = 
\sum_{k=1}^{n} w_k  \vec{\mu}_k.
$$
The covariance matrix of the $k$-th triangle is
$$
\begin{array}{lll}
\Sigma_k & = &
\frac{\int_0^1 \int_0^{1-s} {(\vec{T}_k}(s,t) - \vec{\mu}) \,
(\vec{T}_k(s,t) - \vec{\mu})^T \, d t \, d s}
{\int_0^1 \int_0^{1-s} \, d t \, d s} \\
& = & 
\frac{1}{12} \Big( \sum_{j=1}^{3} \sum_{h=1}^{3}
( {\vec{x}_{j,k}} - \vec{\mu} ) {({\vec{x}_{h,k}} - \vec{\mu} )}^T + \\
& & \quad \;\;\; \sum_{j=1}^{3} ( {\vec{x}_{j,k}} - \vec{\mu} ) {({\vec{x}_{j,k}} - \vec{\mu} )}^T \Big).
\end{array}
$$
The $(i,j)$-th element of $\Sigma_k$, $i, j\in \{1, 2\}$, is 
$$
\begin{array}{lll}
\sigma_{ij, k} & = &
 \frac{1}{12} \Big( \sum_{l=1}^{3} \sum_{h=1}^{3}
( x_{i, l, k} - \mu_i ) (x_{j, h, k} - \mu_j ) 
+ \\
& & \quad \;\;\;\, \sum_{l=1}^{3} ( x_{i, l, k} - \mu_i ) (x_{j, l, k} - \mu_j ) \Big),
\end{array}
$$
with $\vec{\mu}=(\mu_1, \mu_2 )$.
The covariance matrix of $X$ is
$$
\begin{array}{lll}
\Sigma & = &
\sum_{k=1}^{n} w_k \Sigma_k.
\end{array}
$$

\bigskip
\noindent
{\bf Adding points}
\label{subsec:adding-point20-2D}

\bigskip

We add points to $X$.
Let $X'$ be the new convex hull.
We assume that $X'$ is  obtained from $X$ by deleting $n_d$, and
adding $n_a$ triangles.
Then the sum of the areas of all triangles is
$$
a' = \sum_{k=1}^{n} a_k + \sum_{k=1}^{n_a} a_k -\sum_{k=1}^{n_d} a_k
   = a + \sum_{k=1}^{n_a} a_k -\sum_{k=1}^{n_d} a_k.
$$

The center of gravity of $X'$ is
\begin{equation}\label{eq:add100-2D}
\begin{array}{lll}
\vec{\mu}'  & = & \sum_{k=1}^{n}w_k'\vec{\mu}_k +
            \sum_{k=1}^{n_a}w_k'\vec{\mu}_k -
            \sum_{k=1}^{n_d}w_k'\vec{\mu}_k \vspace{0.3cm} \\
      & = & \frac{1}{a'}\left( \sum_{k=1}^{n}a_k\vec{\mu}_k +
            \sum_{k=1}^{n_a}a_k\vec{\mu}_k -
            \sum_{k=1}^{n_d}a_k\vec{\mu}_k \right) \vspace{0.3cm} \\
      & = & \frac{1}{a'}\left( a \vec{\mu} +
            \sum_{k=1}^{n_a}a_k\vec{\mu}_k -
            \sum_{k=1}^{n_d}a_k\vec{\mu}_k\right).
\end{array}
\end{equation}

Let
$$
\vec{\mu}_a = \frac{1}{a'}\sum_{k=1}^{n_a} a_k\vec{\mu}_k, \;\;\;\text{and}\;\;\;
\vec{\mu}_d = \frac{1}{a'}\sum_{k=1}^{n_d} a_k\vec{\mu}_k.
$$

Then, we can rewrite (\ref{eq:add100-2D}) as
\begin{equation}\label{eq:add110-2D}
\vec{\mu}'  =  \frac{a}{a'} \vec{\mu} + \vec{\mu}_a - \vec{\mu}_d.
\end{equation}

The $i$-th component of $\vec{\mu}_a$ and $\vec{\mu}_d$, $1 \leq i \leq 2$,
is denoted by $\mu_{i,a}$ and $\mu_{i,d}$, respectively.
The $(i,j)$-th component, $\sigma_{ij}'$, $1 \leq i, j \leq 2$, of the covariance matrix 
$\Sigma'$ of $X'$ is
$$
\begin{array}{lll}
\sigma_{ij}' & = &
 \frac{1}{12} \Big( \sum_{k=1}^{n} \sum_{l=1}^{3} \sum_{h=1}^{3}
w_k' ( x_{i,l,k} - \mu_i' ) (x_{j, h,k} - \mu_j' ) + \vspace{0.3cm}\\
& & 
\hspace{0.75cm} \sum_{k=1}^{n} \sum_{l=1}^{3} w_k' ( x_{i,l,k}- \mu_i' ) (x_{j,l,k} - \mu_j ') \Big) + \vspace{0.3cm}\\
& &
\frac{1}{12} \Big( \sum_{k=1}^{n_a} \sum_{l=1}^{3} \sum_{h=1}^{3}
w_k' ( x_{i,l,k} - \mu_i' ) (x_{j, h,k} - \mu_j' ) + \vspace{0.3cm}\\
& & 
\hspace{0.75cm} \sum_{k=1}^{n_a} \sum_{l=1}^{3} w_k' ( x_{i,l,k}- \mu_i' ) (x_{j,l,k} - \mu_j ') - \vspace{0.3cm}\\
& &
\hspace{0.75cm}\sum_{k=1}^{n_d} \sum_{l=1}^{3} \sum_{h=1}^{3}
w_k' ( x_{i,l,k} - \mu_i' ) (x_{j, h,k} - \mu_j' ) - \vspace{0.3cm}\\
& & 
\hspace{0.75cm} \sum_{k=1}^{n_d} \sum_{l=1}^{3} w_k' ( x_{i,l,k}- \mu_i' ) (x_{j,l,k} - \mu_j ')\Big).
\end{array}
$$

Let

$$
\sigma_{ij}' = \frac{1}{12} (\sigma_{ij,11}' + \sigma_{ij,12}' + \sigma_{ij,21}' + \sigma_{ij,22}'
               - \sigma_{ij,31}'- \sigma_{ij,32}'),
$$

where,
\begin{equation}\label{eq:add130-2D}
\sigma_{ij,11}' =  \sum_{k=1}^{n} \sum_{l=1}^{3} \sum_{h=1}^{3}
w_k' ( x_{i,l,k} - \mu_i' ) (x_{j, h,k} - \mu_j' ), 
\end{equation}

\begin{equation}\label{eq:add140-2D}
\sigma_{ij,12}' =  \sum_{k=1}^{n} \sum_{l=1}^{3} w_k' ( x_{i,l,k}- \mu_i' ) (x_{j,l,k} - \mu_j '),
\end{equation}

\begin{equation}\label{eq:add150-2D}
\sigma_{ij,21}' =  \sum_{k=1}^{n_a} \sum_{l=1}^{3} \sum_{h=1}^{3}
w_k' ( x_{i,l,k} - \mu_i' ) (x_{j, h,k} - \mu_j' ), 
\end{equation}

\begin{equation}\label{eq:add160-2D}
\sigma_{ij,22}' = \sum_{k=1}^{n_a} \sum_{l=1}^{3} w_k' ( x_{i,l,k}- \mu_i' ) (x_{j,l,k} - \mu_j '),
\end{equation}

\begin{equation}\label{eq:add170-2D}
\sigma_{ij,31}' =  \sum_{k=1}^{n_d} \sum_{l=1}^{3} \sum_{h=1}^{3}
w_k' ( x_{i,l,k} - \mu_i' ) (x_{j, h,k} - \mu_j' ), 
\end{equation}

\begin{equation}\label{eq:add180-2D}
\sigma_{ij,32}' = \sum_{k=1}^{n_d} \sum_{l=1}^{3} w_k' ( x_{i,l,k}- \mu_i' ) (x_{j,l,k} - \mu_j ').
\end{equation}

Plugging-in the values of $\mu_i'$ and $\mu_j'$ in (\ref{eq:add130-2D}), we obtain:

\begin{equation}\label{eq:add190-2D}
\begin{array}{lll}

\sigma_{ij,11}'& = & \sum_{k=1}^{n} \sum_{l=1}^{3} \sum_{h=1}^{3}
w_k' ( x_{i,l,k} - \frac{a}{a'} \mu_i - \mu_{i,a} + \mu_{i,d} )
     (x_{j, h,k} - \frac{a}{a'} \mu_j - \mu_{j,a} + \mu_{j,d} ) \vspace{0.3cm} \\
              & = & \sum_{k=1}^{n} \sum_{l=1}^{3} \sum_{h=1}^{3}
                    w_k' ( x_{i,l,k} - \mu_i + \mu_i(1- \frac{a}{a'}) - \mu_{i,a} + \mu_{i,d} ) \vspace{0.3cm} \\
              & &\hspace{3.7cm}(x_{j, h,k} - \mu_j + \mu_j(1-\frac{a}{a'})  - \mu_{j,a} + \mu_{j,d} ) \vspace{0.3cm} \\
              & = & \sum_{k=1}^{n} \sum_{l=1}^{3} \sum_{h=1}^{3}
                      w_k' ( x_{i,l,k} - \mu_i)(x_{j, h,k} - \mu_j) + \vspace{0.3cm} \\
              & &  \sum_{k=1}^{n} \sum_{l=1}^{3} \sum_{h=1}^{3}
                      w_k' ( x_{i,l,k} - \mu_i)(\mu_j(1-\frac{a}{a'})  - \mu_{j,a} + \mu_{j,d} ) + \vspace{0.3cm} \\
              & & \sum_{k=1}^{n} \sum_{l=1}^{3} \sum_{h=1}^{3}
                      w_k' (\mu_i(1- \frac{a}{a'}) - \mu_{i,a} + \mu_{i,d})(x_{j, h,k} - \mu_j) + \vspace{0.3cm} \\
              & & \sum_{k=1}^{n} \sum_{l=1}^{3} \sum_{h=1}^{3}
                      w_k' (\mu_i(1- \frac{a}{a'}) - \mu_{i,a} + \mu_{i,d})(\mu_j(1-\frac{a}{a'})  - \mu_{j,a} + \mu_{j,d}).
\end{array}
\end{equation}

Since $ \sum_{k=1}^{n} \sum_{l=1}^{3} w_k' ( x_{i,l,k} - \mu_i)=0$, $1 \leq i \leq 2$, we have

\begin{equation}\label{eq:add200-2D}
\begin{array}{lll}
\sigma_{ij,11}' & = & \frac{1}{a'}\sum_{k=1}^{n} \sum_{l=1}^{3} \sum_{h=1}^{3}
                      a_k ( x_{i,l,k} - \mu_i)(x_{j, h,k} - \mu_j) + \vspace{0.3cm} \\
                &   &  \frac{1}{a'} \sum_{k=1}^{n} \sum_{l=1}^{3} \sum_{h=1}^{3}
      a_k (\mu_i(1- \frac{a}{a'}) - \mu_{i,a} + \mu_{i,d})(\mu_j(1-\frac{a}{a'})  - \mu_{j,a} + \mu_{j,d})\vspace{0.3cm} \\
      & = & \frac{1}{a'}\sum_{k=1}^{n} \sum_{l=1}^{3} \sum_{h=1}^{3}
                      a_k ( x_{i,l,k} - \mu_i)(x_{j, h,k} - \mu_j) + \vspace{0.3cm} \\
                &   &  9 \frac{a}{a'}(\mu_i(1- \frac{a}{a'}) - \mu_{i,a} + \mu_{i,d})(\mu_j(1-\frac{a}{a'})  - \mu_{j,a} + \mu_{j,d}).
\end{array}
\end{equation}

Plugging-in the values of $\mu_i'$ and $\mu_j'$ in (\ref{eq:add140-2D}), we obtain:

\begin{equation}\label{eq:add210-2D}
\begin{array}{lll}
\sigma_{ij,12}' & = & \sum_{k=1}^{n} \sum_{l=1}^{3}
w_k' ( x_{i,l,k} - \frac{a}{a'} \mu_i - \mu_{i,a} + \mu_{i,d} )
     (x_{j, h,k} - \frac{a}{a'} \mu_j - \mu_{j,a} + \mu_{j,d} ) \vspace{0.3cm} \\
              & = & \sum_{k=1}^{n} \sum_{l=1}^{3}
                    w_k' ( x_{i,l,k} - \mu_i + \mu_i(1- \frac{a}{a'}) - \mu_{i,a} + \mu_{i,d} ) \vspace{0.3cm} \\
              & &\hspace{2.6cm}(x_{j, h,k} - \mu_j + \mu_j(1-\frac{a}{a'})  - \mu_{j,a} + \mu_{j,d} ) \vspace{0.3cm} \\
              & = & \sum_{k=1}^{n} \sum_{l=1}^{3} 
                      w_k' ( x_{i,l,k} - \mu_i)(x_{j, h,k} - \mu_j) + \vspace{0.3cm} \\
              & &  \sum_{k=1}^{n} \sum_{l=1}^{3} 
                      w_k' ( x_{i,l,k} - \mu_i)(\mu_j(1-\frac{a}{a'})  - \mu_{j,a} + \mu_{j,d} ) + \vspace{0.3cm} \\
              & & \sum_{k=1}^{n} \sum_{l=1}^{3} 
                      w_k' (\mu_i(1- \frac{a}{a'}) - \mu_{i,a} + \mu_{i,d})(x_{j, h,k} - \mu_j) + \vspace{0.3cm} \\
              & & \sum_{k=1}^{n} \sum_{l=1}^{3} 
                      w_k' (\mu_i(1- \frac{a}{a'}) - \mu_{i,a} + \mu_{i,d})(\mu_j(1-\frac{a}{a'})  - \mu_{j,a} + \mu_{j,d}).
\end{array}
\end{equation}

Since $ \sum_{k=1}^{n} \sum_{l=1}^{3} w_k' ( x_{i,l,k} - \mu_i)=0$, $1 \leq i \leq 2$, we have

\begin{equation}\label{eq:add220-2D}
\begin{array}{lll}
\sigma_{ij,12}' & = & \frac{1}{a'}\sum_{k=1}^{n} \sum_{l=1}^{3}
                      a_k ( x_{i,l,k} - \mu_i)(x_{j, h,k} - \mu_j) + \vspace{0.3cm} \\
                &   &  \frac{1}{a'} \sum_{k=1}^{n} \sum_{l=1}^{3}
      a_k (\mu_i(1- \frac{a}{a'}) - \mu_{i,a} + \mu_{i,d})(\mu_j(1-\frac{a}{a'})  - \mu_{j,a} + \mu_{j,d})\vspace{0.3cm} \\
      & = & \frac{1}{a'}\sum_{k=1}^{n} \sum_{l=1}^{3}
                      a_k ( x_{i,l,k} - \mu_i)(x_{j, h,k} - \mu_j) + \vspace{0.3cm} \\
                &   &  3 \frac{a}{a'}(\mu_i(1- \frac{a}{a'}) - \mu_{i,a} + \mu_{i,d})(\mu_j(1-\frac{a}{a'})  - \mu_{j,a} + \mu_{j,d}).
\end{array}
\end{equation}

From (\ref{eq:add210-2D}) and (\ref{eq:add220-2D}), we obtain

\begin{equation}\label{eq:add230-2D}
\begin{array}{lll}
\sigma_{ij,1}' & = & \sigma_{ij,11}' + \sigma_{ij,12}' \vspace{0.3cm}\\
               & = & \sigma_{ij} + 12 \frac{a}{a'}(\mu_i(1- \frac{a}{a'}) - \mu_{i,a} + \mu_{i,d})(\mu_j(1-\frac{a}{a'})  - \mu_{j,a} + \mu_{j,d}).
      \end{array}
\end{equation}

Note that $\sigma_{ij,1}'$ can be computed in $O(1)$ time.
The components $\sigma_{ij,21}'$ and $\sigma_{ij,22}'$ can be computed in $O(n_a)$ time,
while $O(n_d)$ time is needed for computing $\sigma_{ij,31}'$ and  $\sigma_{ij,32}'$.
Thus, $\vec{\mu'}$ and

\begin{equation}\label{eq:add240-2D}
\begin{array}{lll}
\sigma_{ij}' & = & \frac{1}{12}(\sigma_{ij,11}' + \sigma_{ij,12}' + \sigma_{ij,21}' + \sigma_{ij,22}' +
                   \sigma_{ij,31}' + \sigma_{ij,32}') \vspace{0.3cm}\\
               & = & \frac{1}{12}(\sigma_{ij} + \sigma_{ij,21}' + \sigma_{ij,22}' + \sigma_{ij,31}' + \sigma_{ij,32}) +  \vspace{0.3cm}\\ 
               & & \frac{a}{a'}(\mu_i(1- \frac{a}{a'}) - \mu_{i,a} + \mu_{i,d})(\mu_j(1-\frac{a}{a'})  - \mu_{j,a} + \mu_{j,d}).
\end{array}
\end{equation}

can be computed in $O(n_a + n_d)$ time.

\bigskip
\noindent
{\bf {Deleting points}}
\label{subsec:deleting-point20-2D}
\bigskip

Let the new convex hull be obtained by deleting
$n_d$ tetrahedra from and added $n_a$ tetrahedra to the old convex hull. 
If the interior point $\vec{o}$, after deleting points, lies inside the new convex hull, then
the same formulas and time complexity, as by adding points, follow.
However, $\vec{o}$ could lie outside the new convex hull. Then, we need
to choose a new interior point $\vec{o}'$, and recompute the new tetrahedra associated with it.
Thus, we need in total $O(n)$ time to update the principal components.

\subsubsection{Continuous PCA over the boundary of a polygon}
\label{subsec:cpca-polygon-boundary}

Let $X$ be a polygon in $\mathbb{R}^2$.
We assume that the boundary of $X$ is comprised of 
$n$ line segments.
The $k$-th line segment, with vertices 
${\vec{x}_{1,k}}, {\vec{x}_{2,k}}$, can be represented in a parametric form
by
$$
{\vec{S}_k}(t)={\vec{x}_{1,k}} + 
t \, ( {\vec{x}_{2,k}} - {\vec{x}_{1,k}} ). 
$$
Since we assume that the mass density is constant,
the center of gravity of the $k$-th line segment is
$$
\vec{\mu}_k=\frac{\int_0^1  {\vec{S}_k}(t) \, d t}
{\int_0^1  \, d t} =
\frac{{\vec{x}_{1,k}}+ {\vec{x}_{2,k}}}{2}.
$$
The contribution of each line segment to the center of gravity of the boundary of
a polygon is proportional with the length of the line segment.
The length of the $k$-th line segment is
$$
s_k={\mbox{length}}(S_k) =
|| {\vec{x}_{2,k}} - {\vec{x}_{1,k}} ||.
$$
We introduce a weight to each
line segment that is proportional with its length, define as
$$
w_k = \frac{s_k}{\sum_{k=1}^{n} s_k} = \frac{s_k}{s},
$$
where $s$ is the perimeter of $X$. Then, the center of gravity of the boundary of $X$ is
$$
\vec{\mu}  = 
\sum_{k=1}^{n} w_k  \vec{\mu}_k.
$$
The covariance matrix of the $k$-th line segment is
$$
\begin{array}{lll}
\Sigma_k & = &
\frac{\int_0^1  {(\vec{S}_k}(t) - \vec{\mu}) \,
(\vec{S}_k(t) - \vec{\mu})^T \, d t}
{\int_0^1 \, d t} \\
& = & 
\frac{1}{6} \Big( \sum_{j=1}^{2} \sum_{h=1}^{2}
( {\vec{x}_{j,k}} - \vec{\mu} ) {({\vec{x}_{h,k}} - \vec{\mu} )}^T + \\
&  &  \quad \;\, \sum_{j=1}^{2} ( {\vec{x}_{j,k}} - \vec{\mu} ) {({\vec{x}_{j,k}} - \vec{\mu} )}^T \Big).
\end{array}
$$

The $(i,j)$-th element of $\Sigma_k$, $i, j\in \{1, 2\}$, is 
$$
\begin{array}{lll}
\sigma_{ij, k} & = &
 \frac{1}{6} \Big( \sum_{l=1}^{2} \sum_{h=1}^{2}
( x_{i, l, k} - \mu_i ) (x_{j, h, k} - \mu_j ) 
+ \\
& & \quad \;\;\;\, \sum_{l=1}^{2} ( x_{i, l, k} - \mu_i ) (x_{j, l, k} - \mu_j ) \Big),
\end{array}
$$
with $\vec{\mu}=(\mu_1, \mu_2 )$.

The covariance matrix of the boundary of $X$ is
$$
\begin{array}{lll}
\Sigma & = &
\sum_{k=1}^{n} w_k \Sigma_k.
\end{array}
$$

\newpage
\noindent
{\bf Adding points}
\label{subsec:adding-point20-B2D}

\bigskip

We add points to $X$.
Let $X'$ be the new convex hull.
We assume that $X'$ is obtained from $X$ by deleting $n_d$, and
adding $n_a$ line segments.
Then the sum of the lengths of all line segments is
$$
s' = \sum_{k=1}^{n} l_k + \sum_{k=1}^{n_a} s_k -\sum_{k=1}^{n_d} s_k
   = s + \sum_{k=1}^{n_a} s_k -\sum_{k=1}^{n_d} s_k.
$$

The center of gravity of $X'$ is
\begin{equation}\label{eq:add100-B2D}
\begin{array}{lll}
\vec{\mu}'  & = & \sum_{k=1}^{n}w_k'\vec{\mu}_k +
            \sum_{k=1}^{n_a}w_k'\vec{\mu}_k -
            \sum_{k=1}^{n_d}w_k'\vec{\mu}_k \vspace{0.3cm} \\
      & = & \frac{1}{s'}\left( \sum_{k=1}^{n}s_k\vec{\mu}_k +
            \sum_{k=1}^{n_a}s_k\vec{\mu}_k -
            \sum_{k=1}^{n_d}s_k\vec{\mu}_k \right) \vspace{0.3cm} \\
      & = & \frac{1}{s'}\left( s \vec{\mu} +
            \sum_{k=1}^{n_a}s_k\vec{\mu}_k -
            \sum_{k=1}^{n_d} s_k\vec{\mu}_k\right).
\end{array}
\end{equation}

Let
$$
\vec{\mu}_a = \frac{1}{s'}\sum_{k=1}^{n_a} s_k\vec{\mu}_k, \;\;\;\text{and}\;\;\;
\vec{\mu}_d = \frac{1}{s'}\sum_{k=1}^{n_d} s_k\vec{\mu}_k.
$$

Then, we can rewrite (\ref{eq:add100-B2D}) as
\begin{equation}\label{eq:add110-B2D}
\vec{\mu}'  =  \frac{s}{s'} \vec{\mu} + \vec{\mu}_a - \vec{\mu}_d.
\end{equation}

The $i$-th component of $\vec{\mu}_a$ and $\vec{\mu}_d$, $1 \leq i \leq 2$,
is denoted by $\mu_{i,a}$ and $\mu_{i,d}$, respectively.
The $(i,j)$-th component, $\sigma_{ij}'$, $1 \leq i, j \leq 2$, of the covariance matrix 
$\Sigma'$ of $X'$ is
$$
\begin{array}{lll}
\sigma_{ij}' & = &
 \frac{1}{6} \Big( \sum_{k=1}^{n} \sum_{l=1}^{2} \sum_{h=1}^{2}
w_k' ( x_{i,l,k} - \mu_i' ) (x_{j, h,k} - \mu_j' ) + \vspace{0.3cm}\\
& & 
\hspace{0.75cm} \sum_{k=1}^{n} \sum_{l=1}^{2} w_k' ( x_{i,l,k}- \mu_i' ) (x_{j,l,k} - \mu_j ') \Big) + \vspace{0.3cm}\\
& &
\frac{1}{6} \Big( \sum_{k=1}^{n_a} \sum_{l=1}^{2} \sum_{h=1}^{2}
w_k' ( x_{i,l,k} - \mu_i' ) (x_{j, h,k} - \mu_j' ) + \vspace{0.3cm}\\
& & 
\hspace{0.75cm} \sum_{k=1}^{n_a} \sum_{l=1}^{2} w_k' ( x_{i,l,k}- \mu_i' ) (x_{j,l,k} - \mu_j ') - \vspace{0.3cm}\\
& &
\hspace{0.75cm}\sum_{k=1}^{n_d} \sum_{l=1}^{2} \sum_{h=1}^{2}
w_k' ( x_{i,l,k} - \mu_i' ) (x_{j, h,k} - \mu_j' ) - \vspace{0.3cm}\\
& & 
\hspace{0.75cm} \sum_{k=1}^{n_d} \sum_{l=1}^{2} w_k' ( x_{i,l,k}- \mu_i' ) (x_{j,l,k} - \mu_j ')\Big).
\end{array}
$$

Let

$$
\sigma_{ij}' = \frac{1}{6} (\sigma_{ij,11}' + \sigma_{ij,12}' + \sigma_{ij,21}'+ \sigma_{ij,22}'
               - \sigma_{ij,31}'- \sigma_{ij,32}'),
$$

where,
\begin{equation}\label{eq:add130-B2D}
\sigma_{ij,11}' =  \sum_{k=1}^{n} \sum_{l=1}^{2} \sum_{h=1}^{2}
w_k' ( x_{i,l,k} - \mu_i' ) (x_{j, h,k} - \mu_j' ), 
\end{equation}

\begin{equation}\label{eq:add140-B2D}
\sigma_{ij,12}' =  \sum_{k=1}^{n} \sum_{l=1}^{2} w_k' ( x_{i,l,k}- \mu_i' ) (x_{j,l,k} - \mu_j '),
\end{equation}

\begin{equation}\label{eq:add150-B2D}
\sigma_{ij,21}' =  \sum_{k=1}^{n_a} \sum_{l=1}^{2} \sum_{h=1}^{2}
w_k' ( x_{i,l,k} - \mu_i' ) (x_{j, h,k} - \mu_j' ), 
\end{equation}

\begin{equation}\label{eq:add160-B2D}
\sigma_{ij,22}' = \sum_{k=1}^{n_a} \sum_{l=1}^{2} w_k' ( x_{i,l,k}- \mu_i' ) (x_{j,l,k} - \mu_j '),
\end{equation}

\begin{equation}\label{eq:add170-B2D}
\sigma_{ij,31}' =  \sum_{k=1}^{n_d} \sum_{l=1}^{2} \sum_{h=1}^{2}
w_k' ( x_{i,l,k} - \mu_i' ) (x_{j, h,k} - \mu_j' ), 
\end{equation}

\begin{equation}\label{eq:add180-B2D}
\sigma_{ij,32}' = \sum_{k=1}^{n_d} \sum_{l=1}^{3} w_k' ( x_{i,l,k}- \mu_i' ) (x_{j,l,k} - \mu_j ').
\end{equation}

Plugging-in the values of $\mu_i'$ and $\mu_j'$ in (\ref{eq:add130-B2D}), we obtain:

\begin{equation}\label{eq:add190-B2D}
\begin{array}{lll}

\sigma_{ij,11}'& = & \sum_{k=1}^{n} \sum_{l=1}^{2} \sum_{h=1}^{2}
w_k' ( x_{i,l,k} - \frac{s}{s'} \mu_i - \mu_{i,a} + \mu_{i,d} )
     (x_{j, h,k} - \frac{s}{s'} \mu_j - \mu_{j,a} + \mu_{j,d} ) \vspace{0.3cm} \\
              & = & \sum_{k=1}^{n} \sum_{l=1}^{2} \sum_{h=1}^{2}
                    w_k' ( x_{i,l,k} - \mu_i + \mu_i(1- \frac{s}{s'}) - \mu_{i,a} + \mu_{i,d} ) \vspace{0.3cm} \\
              & &\hspace{3.7cm}(x_{j, h,k} - \mu_j + \mu_j(1-\frac{s}{s'})  - \mu_{j,a} + \mu_{j,d} ) \vspace{0.3cm} \\
              & = & \sum_{k=1}^{n} \sum_{l=1}^{2} \sum_{h=1}^{2}
                      w_k' ( x_{i,l,k} - \mu_i)(x_{j, h,k} - \mu_j) + \vspace{0.3cm} \\
              & &  \sum_{k=1}^{n} \sum_{l=1}^{2} \sum_{h=1}^{2}
                      w_k' ( x_{i,l,k} - \mu_i)(\mu_j(1-\frac{s}{s'})  - \mu_{j,a} + \mu_{j,d} ) + \vspace{0.3cm} \\
              & & \sum_{k=1}^{n} \sum_{l=1}^{2} \sum_{h=1}^{2}
                      w_k' (\mu_i(1- \frac{s}{s'}) - \mu_{i,a} + \mu_{i,d})(x_{j, h,k} - \mu_j) + \vspace{0.3cm} \\
              & & \sum_{k=1}^{n} \sum_{l=1}^{2} \sum_{h=1}^{2}
                      w_k' (\mu_i(1- \frac{s}{s'}) - \mu_{i,a} + \mu_{i,d})(\mu_j(1-\frac{s}{s'})  - \mu_{j,a} + \mu_{j,d}).
\end{array}
\end{equation}

Since $ \sum_{k=1}^{n} \sum_{l=1}^{2} w_k' ( x_{i,l,k} - \mu_i)=0$, $1 \leq i \leq 2$, we have

\begin{equation}\label{eq:add200-B2D}
\begin{array}{lll}
\sigma_{ij,11}' & = & \frac{1}{s'}\sum_{k=1}^{n} \sum_{l=1}^{2} \sum_{h=1}^{2}
                      s_k ( x_{i,l,k} - \mu_i)(x_{j, h,k} - \mu_j) + \vspace{0.3cm} \\
                &   &  \frac{1}{s'} \sum_{k=1}^{n} \sum_{l=1}^{2} \sum_{h=1}^{2}
      s_k (\mu_i(1- \frac{s}{s'}) - \mu_{i,a} + \mu_{i,d})(\mu_j(1-\frac{s}{s'})  - \mu_{j,a} + \mu_{j,d})\vspace{0.3cm} \\
      & = & \frac{1}{s'}\sum_{k=1}^{n} \sum_{l=1}^{2} \sum_{h=1}^{2}
                      s_k ( x_{i,l,k} - \mu_i)(x_{j, h,k} - \mu_j) + \vspace{0.3cm} \\
                &   &  4 \frac{s}{s'}(\mu_i(1- \frac{s}{s'}) - \mu_{i,a} + \mu_{i,d})(\mu_j(1-\frac{s}{s'})  - \mu_{j,a} + \mu_{j,d}).
\end{array}
\end{equation}

Plugging-in the values of $\mu_i'$ and $\mu_j'$ in (\ref{eq:add140-B2D}), we obtain:

\begin{equation}\label{eq:add210-B2D}
\begin{array}{lll}
\sigma_{ij,12}' & = & \sum_{k=1}^{n} \sum_{l=1}^{2}
w_k' ( x_{i,l,k} - \frac{s}{s'} \mu_i - \mu_{i,a} + \mu_{i,d} )
     (x_{j, h,k} - \frac{s}{s'} \mu_j - \mu_{j,a} + \mu_{j,d} ) \vspace{0.3cm} \\
              & = & \sum_{k=1}^{n} \sum_{l=1}^{2}
                    w_k' ( x_{i,l,k} - \mu_i + \mu_i(1- \frac{s}{s'}) - \mu_{i,a} + \mu_{i,d} ) \vspace{0.3cm} \\
              & &\hspace{2.6cm}(x_{j, h,k} - \mu_j + \mu_j(1-\frac{s}{s'})  - \mu_{j,a} + \mu_{j,d} ) \vspace{0.3cm} \\
              & = & \sum_{k=1}^{n} \sum_{l=1}^{2} 
                      w_k' ( x_{i,l,k} - \mu_i)(x_{j, h,k} - \mu_j) + \vspace{0.3cm} \\
              & &  \sum_{k=1}^{n} \sum_{l=1}^{2} 
                      w_k' ( x_{i,l,k} - \mu_i)(\mu_j(1-\frac{s}{s'})  - \mu_{j,a} + \mu_{j,d} ) + \vspace{0.3cm} \\
              & & \sum_{k=1}^{n} \sum_{l=1}^{2} 
                      w_k' (\mu_i(1- \frac{s}{s'}) - \mu_{i,a} + \mu_{i,d})(x_{j, h,k} - \mu_j) + \vspace{0.3cm} \\
              & & \sum_{k=1}^{n} \sum_{l=1}^{2} 
                      w_k' (\mu_i(1- \frac{s}{s'}) - \mu_{i,a} + \mu_{i,d})(\mu_j(1-\frac{s}{s'})  - \mu_{j,a} + \mu_{j,d}).
\end{array}
\end{equation}

Since $ \sum_{k=1}^{n} \sum_{l=1}^{2} w_k' ( x_{i,l,k} - \mu_i)=0$, $1 \leq i \leq 2$, we have

\begin{equation}\label{eq:add220-B2D}
\begin{array}{lll}
\sigma_{ij,12}' & = & \frac{1}{s'}\sum_{k=1}^{n} \sum_{l=1}^{2}
                      s_k ( x_{i,l,k} - \mu_i)(x_{j, h,k} - \mu_j) + \vspace{0.3cm} \\
                &   &  \frac{1}{s'} \sum_{k=1}^{n} \sum_{l=1}^{2}
      s_k (\mu_i(1- \frac{s}{s'}) - \mu_{i,a} + \mu_{i,d})(\mu_j(1-\frac{s}{s'})  - \mu_{j,a} + \mu_{j,d})\vspace{0.3cm} \\
      & = & \frac{1}{s'}\sum_{k=1}^{n} \sum_{l=1}^{2}
                      s_k ( x_{i,l,k} - \mu_i)(x_{j, h,k} - \mu_j) + \vspace{0.3cm} \\
                &   &  2 \frac{s}{s'}(\mu_i(1- \frac{s}{s'}) - \mu_{i,a} + \mu_{i,d})(\mu_j(1-\frac{s}{s'})  - \mu_{j,a} + \mu_{j,d}).
\end{array}
\end{equation}

From (\ref{eq:add210-B2D}) and (\ref{eq:add220-B2D}), we obtain

\begin{equation}\label{eq:add230-B2D}
\begin{array}{lll}
\sigma_{ij,1}' & = & \sigma_{ij,11}' + \sigma_{ij,12}' \vspace{0.3cm}\\
               & = & \sigma_{ij} + 6 \frac{s}{s'}(\mu_i(1- \frac{s}{s'}) - \mu_{i,a} + \mu_{i,d})(\mu_j(1-\frac{s}{s'})  - \mu_{j,a} + \mu_{j,d}).
      \end{array}
\end{equation}

Note that $\vec{\mu}'$ and the components $\sigma_{ij,21}'$, $\sigma_{ij,22}'$, $\sigma_{i,31}'$ and 
$\sigma_{i,32}'$ can be computed in
constant time, under the assumption that $X$ and $X'$ differ in 
the constant number of polyhedra, i.e., $n_a$ and $n_d$ are constants.
In that case, also the component 

\begin{equation}\label{eq:add240-B2D}
\begin{array}{lll}
\sigma_{ij}' & = & \frac{1}{6}(\sigma_{ij,11}' + \sigma_{ij,12}' + \sigma_{ij,21}' + \sigma_{ij,22}' +
                   \sigma_{ij,31}' + \sigma_{ij,32}') \vspace{0.3cm}\\
               & = & \frac{1}{6}(\sigma_{ij} + \sigma_{ij,21}' + \sigma_{ij,22}' + \sigma_{ij,31}' + \sigma_{ij,32}) +  \vspace{0.3cm}\\ 
               & & \frac{s}{s'}(\mu_i(1- \frac{s}{s'}) - \mu_{i,a} + \mu_{i,d})(\mu_j(1-\frac{s}{s'})  - \mu_{j,a} + \mu_{j,d})
\end{array}
\end{equation}

can be computed in $O(n_a + n_d)$ time.

\bigskip
\noindent
{\bf {Deleting points}}
\label{subsec:deleting-point20-2D-B2D}
\bigskip

Let the new convex hull be obtained by deleting
$n_d$ tetrahedra from and added $n_a$ tetrahedra to the old convex hull. 
Consequently,
the same formulas and time complexity, as by adding points, follow.

\end{document}